\pdfoutput=1
\documentclass[journal]{IEEEtran}
\usepackage[utf8]{inputenc}
\usepackage{amsmath}
\usepackage{amsfonts}
\usepackage{bbding}
\usepackage{amssymb}
\usepackage{array}
\usepackage{multirow}
\usepackage{graphicx}
\usepackage[caption=false,font=footnotesize]{subfig}
\usepackage{algorithm}
\usepackage{algorithmic}
\usepackage[hidelinks]{hyperref}
\newcommand\algorithmicprocedure{\textbf{procedure}}
\newcommand{\algorithmicendprocedure}{\algorithmicend\ \algorithmicprocedure}
\makeatletter
\newcommand\PROCEDURE[3][default]{%
  \ALC@it
  \algorithmicprocedure\ \textsc{#2}(#3)%
  \ALC@com{#1}%
  \begin{ALC@prc}%
}
\newcommand\ENDPROCEDURE{%
  \end{ALC@prc}%
  \ifthenelse{\boolean{ALC@noend}}{}{%
    \ALC@it\algorithmicendprocedure
  }%
}
\newenvironment{ALC@prc}{\begin{ALC@g}}{\end{ALC@g}}
\makeatother

\usepackage{xcolor}
\usepackage{cite}

\allowdisplaybreaks

\DeclareMathOperator{\gram}{\bf gram}

\DeclareMathOperator{\Tr}{\bf Tr}
\DeclareMathOperator{\nsl}{\bf nsl}

\usepackage{amsthm}
\usepackage{stfloats}
\newtheorem{defi}{Definition}
\newtheorem{theorem}{Theorem}

\newtheorem{rem}{Remark}

\newcommand{\blue}[1]{\textcolor{blue}{#1}}

\begin{document}
\title{SIG-SDP: Sparse Interference Graph-Aided Semidefinite Programming for Large-Scale Wireless Time-Sensitive Networking}

\author{
\IEEEauthorblockN{Zhouyou Gu, Jihong Park, Branka Vucetic, Jinho Choi}\\
\thanks{
This work was supported in part by A$^*$STAR under its IAF-ICP (I2501E0064), in part by the IITP-ITRC grant funded by the Korean government (MSIT) (IITP-2026-RS-2023-00259991) (33\%), in part by SUTD Kickstarter Initiative (SKI 2021 06 08), and in part by the National Research Foundation, Singapore, and the Infocomm Media Development Authority under its Future Communications Research \& Development Programme. The work was supported in part by the Australian Research Council Laureate Fellowship grant number FL160100032 and Discovery grant number DP210103410.
\emph{(Corresponding authors: J. Park and Z. Gu.)}}
\thanks{
Z. Gu and J. Park are with the Information Systems Technology and Design Pillar, Singapore University of Technology and Design, Singapore 487372 (email: \{zhouyou\_gu, jihong\_park\}@sutd.edu.sg).
}
\thanks{
B. Vucetic is with the School of Electrical and Computer Engineering, the University of Sydney, Sydney, NSW 2006, Australia 
(email: \{branka.vucetic\}@sydney.edu.au).
}
\thanks{
J. Choi is with the School of Electrical and Mechanical Engineering,
the University of Adelaide, Adelaide, SA 5005, Australia
(email: \{jinho.choi\}@adelaide.edu.au).
}
\thanks{Source codes are available at {https://github.com/zhouyou-gu/sig-sdp-mmw}.}
\vspace{-0.75cm}
}
\maketitle

\begin{abstract}
Wireless time-sensitive networking (WTSN) is essential for Industrial Internet of Things. We address the problem of minimizing time slots needed for WTSN transmissions while ensuring reliability subject to interference constraints---an NP-hard task. Existing semidefinite programming (SDP) methods can relax and solve the problem but suffer from high polynomial complexity. We propose a sparse interference graph-aided SDP (SIG-SDP) framework that exploits the interference's sparsity arising from attenuated signals between distant user pairs. First, the framework utilizes the sparsity to establish the upper and lower bounds of the minimum number of slots and uses binary search to locate the minimum within the bounds. Here, for each searched slot number, the framework optimizes a positive semidefinite (PSD) matrix indicating how likely user pairs share the same slot, and the constraint feasibility with the optimized PSD matrix further refines the slot search range. Second, the framework designs a matrix multiplicative weights (MMW) algorithm that accelerates the optimization, achieved by only sparsely adjusting interfering user pairs' elements in the PSD matrix while skipping the non-interfering pairs. We also design an online architecture to deploy the framework to adjust slot assignments based on real-time interference measurements. Simulations show that the SIG-SDP framework converges in near-linear complexity and is highly scalable to large networks. The framework minimizes the number of slots with up to 10 times faster computation and up to 100 times lower packet loss rates than compared methods. The online architecture demonstrates how the algorithm complexity impacts dynamic networks' performance.

\end{abstract}


\section{Introduction}
Emerging Industrial Internet of Things (IIoT) applications, including autonomous vehicles, factory automation and tactile internet, have driven the evolution of wireless networks \cite{sisinni2018industrial,vitturi2019industrial,hazra2021comprehensive}.
In typical IIoT scenarios, as illustrated in Fig. \ref{fig:WTSN_factory}, multiple users (e.g., sensors) periodically transmit status updates to base stations (BSs) via the uplink of the wireless network \cite{nasrallah2018ultra}. These status updates follow deterministic packet arrival processes and demand stringent quality-of-service (QoS) requirements, including low latency and high reliability. Delayed or lost sensor updates can lead to inaccurate control decisions, potentially resulting in accidents \cite{3gpp.22.104}.
For example, the missing sensor update on the chemical process temperature and pressure can cause the process to run out of control, leading to explosions and other catastrophic events \cite{zhao2024deep}; Similarly, the delayed/lost measurements on automated guided vehicles' surrounding environment can cause the vehicles to run into obstacles, leading to collisions and damages \cite{bostelman2014methods}.
Wireless network technologies that support such traffic are referred to as wireless time-sensitive networking (WTSN) \cite{zanbouri2024comprehensive}.

\begin{figure}[!t]
\centering
\includegraphics[scale=0.45]{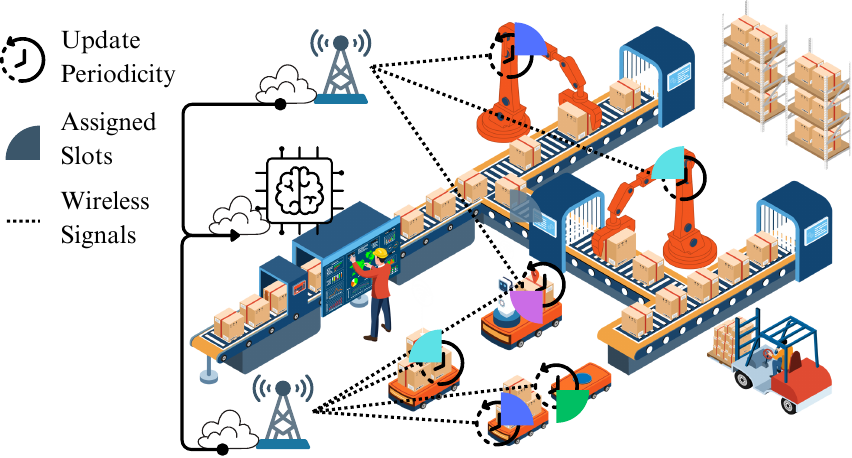}
\vspace{-0.1cm}
\caption{Illustration of a WTSN system for IIoT applications.}
\label{fig:WTSN_factory}
\vspace{-0.5cm}
\end{figure}

Existing WTSN approaches assign a separate time slot to each transmission to ensure transmission reliability \cite{khoshnevisan20195g,seijo2020w,sudhakaran2022wireless,candell2023scheduling,gu2021knowledge,yang2023detfed,zhou2024predictable}. 
Clearly, this slot assignment scheme will cause significant delays in large-scale networks because each user needs to wait for many other users' slots to end before transmitting.
To reduce the delay to the next available slot, multiple users can share the same time slot. However, slot sharing inevitably causes inter-user interference in their transmissions, which decreases signal-to-interference-plus-noise ratios (SINR) and further increases decoding failures at receiving BSs, dramatically degrading the reliability of WTSN. Thus, the challenge arises to find a slot assignment where concurrent user transmissions can be reliably received while minimizing the number of slots that users need to wait in WTSN.

Given that the above slot assignments in WTSN are discrete integers, the task of finding the optimal decisions can be formulated as an integer programming problem \cite{chang2009multicell,liang2018graph}.
Though the problem is NP-hard, it can be relaxed as the one solvable within polynomial time complexity (w.r.t. the number of users) by approximating the integer decisions as continuous numbers \cite{subramanian2008minimum,gu2024graph}. For example, the integer decisions can be relaxed to a positive semidefinite (PSD) matrix whose values represent how likely each pair of users will be assigned to the same slot (i.e., how likely interference occurs between them). This reformulates the interference management task as a semidefinite programming (SDP) problem whose optimal solution (the optimal PSD matrix) can be rounded back to near-optimal assignments \cite{gu2024graph,goemans1995improved}.
However, solving SDP remains time-consuming, especially with a large number of users in the network, leading to delayed solutions \cite{shi2015largeA}.
As a result, the solvers ignore users' increasing interference during the solving time, e.g., when users move close to other users, causing low SINRs and decoding failures.
Therefore, reducing the time spent on solving the interference management problem is crucial for ensuring the high reliability in WTSN.
\subsection{Related Works}\label{subsec:related_works}

\subsubsection{Interference-Free Time-Slot Assignment Methods}
Authors in \cite{khoshnevisan20195g,seijo2020w,sudhakaran2022wireless,candell2023scheduling} implement prototypes of WTSN and manually assign separate periodical slots to each user's transmissions when experimenting with their prototypes. 
To automate the slot assignments, machine learning algorithms are used \cite{gu2021knowledge,yang2023detfed,zhou2024predictable} to allocate separate slots for each user's transmissions. 
None of these works study how to manage interference in WTSN when multiple users inevitably share the same slot as the number of users increases. 
Works in \cite{tassiulas1990stability,joo2009understanding,birand2011analyzing,tsanikidis2021power} formulate the interference in the network as an undirected binary conflict graph, in which users are vertices and an edge connects two users if they are interfering with each other and cannot transmit together.
Users scheduled in each slot are heuristically selected as the independent set of the graph, i.e., any two users in the same slot are not connected in the graph \cite{tassiulas1990stability}. 
However, the binary conflict graph does not model the accumulative interference power from neighboring users.
Note that interfering users can simultaneously transmit provided that their SINRs are high enough for successful decoding, which has better transmission efficiency than separating all interfering users.

\subsubsection{Interference-Tolerant Time-Slot Assignment Methods}
Works in \cite{chang2009multicell,liang2018graph} use weighted graphs to model interference between user pairs in the network, where undirected edges \cite{chang2009multicell} assume the interference equally impacts both users, while directed edges \cite{liang2018graph} differentiate the interference made by two users.
Then, the graph cut divides users into a given number of slots by maximizing the sum of edge weights disconnected between slots. To cut the graph efficiently, the linear and the SDP relaxation can relax the integer slot assignment decisions into continuous variables \cite{subramanian2008minimum,gu2024graph}. 
Specifically, the linear one transforms whether each user is assigned to an available slot as a probability, assuming that users' decisions are independent. Meanwhile, the SDP one transforms whether each pair of users is assigned to the same slot as the correlation between these two users' decisions (i.e., as the PSD matrix), better expressing the likelihood that interference occurs.
Therefore, the SDP relaxation returns a closer approximation and better assignments than the linear one \cite{subramanian2008minimum,gu2024opportunistic}.
However, the above relaxation-based methods \cite{subramanian2008minimum,gu2024graph} assume a fixed number of slots and do not apply to minimizing the status update periodicity in WTSN. Further study is needed on relaxation-based slot assignment methods that minimize the number of slots while ensuring reliable transmissions in WTSN.

Additionally, deep learning \cite{eisen2020optimal,shen2020graph} and spectral graph \cite{zha2001spectral,vandam2016new} methods are also widely applied to solve graph problems. However, deep learning methods, e.g., using graph neural networks (GNNs), have limited expressive power in solving combinatorial problems, particularly in direct approximation of interference graph cut/coloring decisions \cite{xu*2018how,loukas2019what,gu2024graph}.
Meanwhile, spectral graph methods approximate the optimal cut/coloring decision by solving the eigenvalue problem of a single graph's adjacency or Laplacian matrix, which can hardly express extra interference constraints.

\subsubsection{SDP Solvers in Wireless Network Optimizations}
Existing wireless network optimizations \cite{shi2015largeA,liu2024survey} commonly use primal-dual interior point (PDIP) \cite{toh1999sdpt3,yamashita2003implementation} or alternating direction methods of multipliers (ADMM) \cite{o2016conic} algorithms for SDP. 
These solvers iteratively update primal and dual optimization variables to satisfy Karush–Kuhn–Tucker (KKT) conditions, certifying optimality \cite{boyd2004convex}. 
Each iteration's update direction is determined by solving a linear system approximating the KKT conditions at iterated variables. Note that this first-order approximation does not guarantee that the iterated primal variables remain as positive semidefinite. Thus, these solvers require a projection of the primal variables onto a PSD matrix in every iteration. 
Here, solving linear systems and projecting PSD matrices involve matrix decomposition, e.g., LU and eigen decomposition, accounting for the major complexity of these methods \cite{toh1999sdpt3,yamashita2003implementation,o2016conic}. 
Unfortunately, their complexity scales in polynomials of the number of constraints and the PSD matrix size \cite{strang2006linear} (i.e., the number of users). As a result, the above solvers cost a significant amount of time when solving the optimization problem of large-scale wireless networks \cite{shi2015largeA}.
Alternative low-complexity SDP solvers for large-scale WTSN optimizations are needed.

\subsection{Our Contributions}
This paper studies how to manage interference in WTSN using SDP methods that assign time slots for users' transmissions. 
We consider a WTSN system where users periodically transmit a packet in one of the slots in each period. 
We define interference constraints enforcing 1) each BS decodes only one of its associated users' transmissions in each slot and 2) that each user's SINR is to be higher than a certain threshold to ensure a low decoding error rate, where associations and SINRs are both determined by user-to-BS path gains.
The objective is to minimize the number of slots in a period while satisfying the above constraints.
Interference graphs are constructed based on path gains, where the binary and weighted edges between users represent coincident user associations and interference powers, respectively. The sparsity of these graphs is analyzed and exploited to accelerate the proposed methods below.
We design an SDP relaxation that represents whether two users share the same slot as a PSD matrix, relaxing the constraints.
The framework uses binary search to find the minimum number of slots, i.e., increasing the slot number if the constraints are feasible or decreasing it otherwise.
For each searched slot number, the framework first optimizes the PSD matrix by solving an SDP constraint satisfaction problem (CSP) with relaxed constraints and then rounds the solution back to integer assignments.
We format the primal and dual of the SDP CSP and implement the matrix multiplicative weights (MMW) algorithm \cite{arora2007combinatorial,steurer2010fast,carmon2019rank} as its solver.
Specifically, the MMW is designed as a two-player zero-sum game: a player uses the hedge rule to adjust dual variables (constraint weights), maximizing the weighted constraint violations; in contrast, the other player uses the matrix exponential of constraint coefficient matrices to approximate the primal variables (the PSD matrix), minimizing the weighted violations.
We show that the duality gap converges in the MMW for the interference management task.
The deployment of the proposed methods in dynamic networks is studied as well.

Our contributions in this work are listed as follows: 
\begin{itemize}
    \item To the best of our knowledge, this work proposes the first sparse interference graph-aided SDP (SIG-SDP) framework that exploits the sparsity of interference graphs to reduce the SDP complexity for interference management.
    The sparsity arises from the limited number of interfering neighbors for users due to unmeasurable, far-distanced signals.
    Specifically, we derive a reduced search range for the minimum number of slots based on the chromatic numbers of the sparse graphs.
    Moreover, this sparsity enables the exclusion of non-interfering user pairs when solving SDP, e.g., using the MMW.
    Simulations demonstrate that the framework exhibits linear computing time w.r.t. the number of users, i.e., they are highly scalable for large networks, and are up to $10$ times faster than the SDP methods that disregard sparsity \cite{toh1999sdpt3,yamashita2003implementation,o2016conic}.
    \item We design the SDP relaxation on the NP-hard problem that minimizes the status update periodicity while satisfying both association and reliability constraints in WTSN. This relaxation differs from existing relaxation methods \cite{subramanian2008minimum,gu2024graph}, which are applicable only to fixed periodicity scenarios. Simulations show that solving the proposed SDP relaxation problem provides slot assignments with up to 10 times fewer packet losses compared to heuristics \cite{tassiulas1990stability} and up to 100 times fewer packet losses compared to the linear relaxation \cite{subramanian2008minimum}.
    \item 
    We implement the MMW algorithm to solve the interference management SDP task, creating a new alternative SDP solver for wireless network optimizations other than the PDIP \cite{toh1999sdpt3,yamashita2003implementation} and ADMM \cite{o2016conic} solvers.
    We derive the convergence of the duality gap in the MMW, w.r.t., the number $C$ of interference constraints, and the number  $K$ of users. Specifically, the convergence error of the gap is $\mathcal{O}{(\eta K)}$ in $\mathcal{O}{(\frac{1}{\eta^2}(\ln C+\ln K)})$ iterations, where $\eta$ is the step size in the MMW. Simulations show that the gap converges at $\sim100$ iterations in the MMW with a properly configured $\eta$, regardless of the network sizes.
    \item
    We design an architecture that deploys the SIG-SDP framework to adjust slot assignments in WTSN based on online-measured path gains. The SIG-SDP framework samples path gains and finds new slot assignments in parallel with the network while users in the network transmit using previous assignments. Once the framework returns new assignments, the users transmit in the newly assigned slots, and the framework restarts the optimization process.
    Simulations show that the proposed framework achieves up to $5$ times lower packet error rates than the low-complexity heuristic for slow-moving users, while the out-performance decreases as the user speed and interference variability increase. This demonstrates the impact of the algorithm complexity on interference management performance in dynamic networks.
\end{itemize}

\subsection{Notation and Paper Organization}
The $i$-th element of a vector, $\mathbf{x}$, is denoted as $x_{i}$.
The $j$-th element of the $i$-th row of a matrix, $\mathbf{X}$, is denoted as $X_{i,j}$.
We write the definition of elements in a matrix $\mathbf{X}$ as $\mathbf{X}\triangleq[X_{i,j}|X_{i,j} = (\cdots)]$, where $(\cdots)$ is the expression defining the elements in $\mathbf{X}$.
We denote a $K\times K$ positive semidefinite matrix $\mathbf{X}$ as $\mathbf{X} \succeq 0$ (without explicit proof, all PSD matrices in this work are symmetric). 
The inner product of two vectors is denoted as $\langle\mathbf{x},\mathbf{y}\rangle$.
The inner product of two matrices are denoted as $\mathbf{X}\bullet\mathbf{Y} = \sum_i \sum_j X_{i,j} Y_{i,j}$. 
We denote the maximum and minimum eigenvalues of a matrix $\mathbf{X}$ as $\lambda_{\max}(\mathbf{X})$ and $\lambda_{\min}(\mathbf{X})$, respectively.
The $\ell_1$ and $\ell_2$ norm of a vector $\mathbf{x}$ is denoted as $|\mathbf{x}|$ and $\|\mathbf{x}\|$, respectively.
The spectral norm of a matrix $\mathbf{X}$ is denoted as $\|\mathbf{X}\|$, which equals the maximum absolute value of eigenvalues of $\mathbf{X}$ when $\mathbf{X}$ is symmetric.
$\mathbb{I}^{K}$ is the $K\times K$ identity matrix. 
$\Tr(\mathbf{X})$ is the trace of $\mathbf{X}$, i.e., $\Tr(\mathbf{X})=\sum_k X_{k,k}$.
A gram matrix form of a PSD matrix $\mathbf{X}$ is denoted as $\gram(\mathbf{X}) \triangleq   [\mathbf{v}_1,\dots,\mathbf{v}_K]^{\rm T}$ where $\mathbf{X} =[\mathbf{v}_1,\dots,\mathbf{v}_K]^{\rm T}  [\mathbf{v}_1,\dots,\mathbf{v}_K]$. 

The rest of this paper is organized as follows.
Section \ref{sec:system_model} presents the WTSN system model. Section \ref{sec:graph_theory_analysis} defines the sparse interference graphs and explains the overall concept of the SIG-SDP framework. Sections \ref{sec:sdp_relaxation_framework}, \ref{sec:mmw_interference_management}  and \ref{sec:online_architecture} present the implementation of the SDP relaxation, the MMW-based SDP solver and the online architecture of the SIG-SDP framework. Finally, Section \ref{sec:simulation_results} shows the simulations that evaluated the proposed methods, and Section \ref{sec:conclusion} concludes this work.

\section{System Model and Problem Formulation}\label{sec:system_model}
This section presents the WTSN system model and the slot assignment problem for interference management.
\subsection{System Model}\label{subsec:tsn_system_model}
\begin{figure}[!t]
\centering
\includegraphics[scale=0.7]{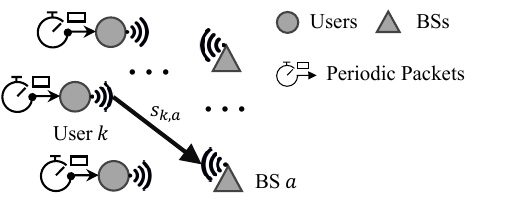}
\vspace{-0.1cm}
\caption{Illustration of a WTSN system.}
\label{fig:system_model_wtsn}
\vspace{-0.4cm}
\end{figure}
We consider a WTSN system, e.g., based on a 5G radio access network consisting of $K$ users and $A$ BSs, as illustrated in Fig. \ref{fig:system_model_wtsn}. 
The BSs are coordinated through a centralized controller and forms a virtual cell \cite{bjornson2020making}, where the users' information is maintained at the central controller and is shared among the BSs once the users join the network.
Here, we assume that users are sensors continuously collecting data to monitor critical information that must be reported to the BSs timely.
Users return the data samples to BSs in the uplink encapsulated in a specific format with a constant length of $L$ in bits.
All users and BSs are assumed to operate in the same channel with a bandwidth $B$ in Hertz (Hz).
We denote the noise power spectral density as $\mathbb{N}_\mathrm{0}$ in Watts/Hz.
All users and BSs are assumed to be synchronized in time, and the time is divided into slots indexed by $t$, where $t = 1, 2, \dots$. The duration of each slot is denoted as $\Delta_\mathrm{0}$ seconds.
Since the sample delay is the processing time at the user and the BS plus the transmission duration of the packet. Then, to minimize the sample delay, we consider a time-division multiple access (TDMA) scheme \cite{rappaport2024wireless} where users' packets are spread over the available spectrum for given time-frequency resources, i.e., minimizing the packet transmission time\footnote{In this work, the slots are considered as time slots, but it can be extended to minimize other orthogonal slots in frequency or code domains subject to interference constraints.}.
User transmissions are configured by semi-persistent scheduling \cite{jiang2007principle} with the same period $Z$. 
Users are divided into $Z$ slots in the period, e.g., $Z$ is a positive integer as
\begin{equation}\label{eq:const:Z_is_positive_integer}
\begin{aligned}
Z\in\mathbb{Z}.
\end{aligned}
\end{equation}
Users in different slots transmit their packets separately in the period.
Specifically, we denote the slots of user $1,\dots,K$ as $\mathbf{z}\triangleq[z_1,\dots,z_K]^{\rm T}$, where $z_k$ is user $k$'s slot satisfying
\begin{equation}\label{eq:const:z_k_in_Z_original}
\begin{aligned}
z_k \in\{1,\dots,Z\},\ \forall k.
\end{aligned}
\end{equation}
Users in the $z$-th slot are $\{k|z_k=z\}$, and they transmit only in periodical time slots $t''-1+z,t''-1+z+Z,t''-1+z+2Z,\dots$, where $t''$ is a non-negative integer indicating the starting slot of the assignments. Here, we assume that the network is static, and the slot assignments repeat and remain the same until the network terminates. Meanwhile, the dynamic allocation of slots will be discussed in Section \ref{sec:online_architecture}, where the slot assignments remains the same in a duration, $[t'',t')$, until the controller updates the assignments at some $t'$.
As a result, the interference only happens among users within the same slot and is eliminated between any two slots.
The periodicity $Z$ directly affects the freshness of status updates from users. A higher value of $Z$ increases delays in status updates, whereas a lower value of $Z$ results in more users sharing the same slots, leading to higher interference that can degrade reliability.

We denote the path gain in decimal from the $k$-th user to the $a$-th BS as $g_{k,a}$, $\forall k, a$. 
Each user is associated with and transmits its packets to the BS with the maximum received signal strength of the user.
Specifically, let $\hat{a}_k$ be the BS that each user $k$ is associated with, i.e., $\hat{a}_k \triangleq \arg \max_a g_{k,a}$, $\forall k$.
We assume that each BS can decode only one user's transmission in each slot.
The decoding error rate of a packet transmission can be estimated based on the SINR $\phi$ as \cite{yang2014quasi}
\begin{equation}\label{eq:bler}
\begin{aligned}
\epsilon (\phi) \approx f_Q\Bigg(\frac{  -L\ln{2}+ {\Delta_\mathrm{0} B}\ln(1+ \phi)}{\sqrt{\Delta_\mathrm{0} B (1- {1}/{[1+\phi]^2)}}}\Bigg),
\end{aligned}
\end{equation}
where $f_Q$ is the tail distribution function of the standard normal distribution. It can be verified that the decoding error rate is monotonically decreasing and increasing with regard to the SINR and the interference power, respectively. We denote the SINR of each user $k$ as $\phi_{k}$, and the decoding error rate of user $k$'s packet is approximated as $\epsilon (\phi_{k})$ using \eqref{eq:bler}.

The WTSN system requires reliable transmissions such that the decoding error rate of each user's packets is less than a threshold $\epsilon^{\max}$. Due to the monotonicity of the error rate in \eqref{eq:bler}, the SINR $\phi_{k}$ of each user $k$ needs to be larger than a threshold $\hat{\phi}$, where $\epsilon (\hat{\phi})=\epsilon^{\max}$.x
To achieve this, user $k$'s transmission power is configured as
\begin{equation}\label{eq:user_tx_power}
\begin{aligned}
P_k \triangleq  (1+\alpha) \hat{\phi} B\mathbb{N}_\mathrm{0} / g_{k, \hat{a}_k} ,\  \forall k ,
\end{aligned}
\end{equation}
where the received signal power at the associated AP is $(1+\alpha) \hat{\phi} B\mathbb{N}_\mathrm{0}$ and $\alpha>0$ indicates the additional transmission power ratio allowing small interference during transmissions while ensuring reliability.
We assume that a dense deployment of BSs is available and each user will always have a BS nearby. This allows each user transmits within its maximum power while ensuring that the received signal power at the associated BS is larger than $(1+\alpha) \hat{\phi} B\mathbb{N}_\mathrm{0}$, i.e., the channel capacity is large enough to transmit $L$ bits in a slot.
The power configuration in \eqref{eq:user_tx_power} leads to the SINR $\phi_{k}$ for each user $k$ as
\begin{equation}\label{eq:sinr}
\begin{aligned}
\phi_{k} 
&\triangleq \frac{P_kg_{k,\hat{a}_k}}{\sum_{k'\neq k}P_{k'} g_{k',\hat{a}_k}\mathbf{1}_{\{z_k = z_{k'}\}}  + B\mathbb{N}_\mathrm{0}} \\
&= \frac{(1+\alpha) \hat{\phi} B\mathbb{N}_\mathrm{0}}{\sum_{k'\neq k}P_{k'} g_{k',\hat{a}_k}\mathbf{1}_{\{z_k = z_{k'}\}} + B\mathbb{N}_\mathrm{0}} ,\ \forall k ,
\end{aligned}
\end{equation}
where $\sum_{k'\neq k}P_{k'} g_{k',\hat{a}_k}\mathbf{1}_{\{z_k = z_{k'}\}}$ is the total interference power from all other users $k'\neq k$ transmitting in the same slot as user $k$.
Since the SINR $\phi_{k}$ is required to be larger than $\hat{\phi}$, we can rearrange \eqref{eq:sinr} to obtain the interference power constraint for user $k$ as
\begin{equation}\label{eq:const:max_interference}
\begin{aligned}
&\phi_{k} \geq \hat{\phi} \Rightarrow \frac{(1+\alpha) \hat{\phi} B\mathbb{N}_\mathrm{0}}{\sum_{k'\neq k}P_{k'} g_{k',\hat{a}_k}\mathbf{1}_{\{z_k = z_{k'}\}}  + B\mathbb{N}_\mathrm{0}} \geq \hat{\phi} \\
&\Rightarrow
  \sum_{k'\neq k}P_{k'} g_{k',\hat{a}_k}\mathbf{1}_{\{z_k = z_{k'}\}}  \leq \alpha\hat{\phi} B\mathbb{N}_\mathrm{0}, \ \forall k .
\end{aligned}
\end{equation}
This means that the maximum interference power a user can tolerate depends on the additional transmission power ratio $\alpha$.

\subsection{Network States}
We collect the binary indicators on whether any two users are associated with the same BS in a $K\times K$ matrix as
\begin{equation}\label{eq:asso_matrix}
\begin{aligned}
\mathbf{Q} \triangleq [Q_{i,j}|Q_{i,j}=\mathbf{1}_{\{\hat{a}_i=\hat{a}_j\}}, \forall i\neq j; Q_{i,i}=0, \forall i] \ ,
\end{aligned}
\end{equation}
where $Q_{i,j}=1$ if users $i$ and $j$ are associated with the same BS (the closest BSs of two users are the same) or otherwise $Q_{i,j}=0$. 
Further, we consider BSs to have practical receiver sensitivity. When the receiving signal strength of user $k$ at BS $a$ is weak, e.g., $P_kg_{k,a}<\gamma B\mathbb{N}_\mathrm{0}$, the user's transmissions cannot be detected and received at the BS, i.e., this interference power value cannot be measured from the central controller's perspective.
Here, $\gamma$ indicates the threshold of the signal power that can be detected and measured. Note that those weak interference still exist in the network and impact the transmissions, while the system cannot measure them.
The measured interference power from the $k$-th user to the $a$-th BS is
\begin{equation}\label{eq:measurable_path_loss}
\begin{aligned}
s_{k,a} \triangleq 
\begin{cases}
    P_kg_{k,a}/(B\mathbb{N}_\mathrm{0}) ,\ \text{if } P_kg_{k,a}\geq \gamma B\mathbb{N}_\mathrm{0} , \\
    0,\ \text{if } P_kg_{k,a}<\gamma B\mathbb{N}_\mathrm{0},
\end{cases} \forall k,a,
\end{aligned}
\end{equation}
where we normalize the interference powers against the noise power if it is measurable, or otherwise set the interference powers as $0$.
All measured interference powers are collected and available at the controller when deciding the slot assignments.
We collect measured interference powers from one user to another user's BS in a $K\times K$ matrix as
\begin{equation}\label{eq:gain_matrix}
\begin{aligned}
\mathbf{S} \triangleq [S_{i,j}|S_{i,j}=s_{i, \hat{a}_j},  \forall i\neq j ;S_{i,i}=0,\forall i ].
\end{aligned}
\end{equation}
Here, $\mathbf{Q}$ and $\mathbf{S}$ are referred to as the network states.

\subsection{Interference Management Problem in WTSN}\label{subsec:im_problem_formulation}
Two users associated with the same BS must be assigned to two different slots to prevent simultaneous transmissions, as each BS can decode only one user per slot, i.e.,
\begin{equation}\label{eq:const:same_bs_diff_slot_original}
\begin{aligned}
\mathbf{1}_{\{z_k = z_{k'}\}} \leq 0, \ \forall Q_{k,k'} =1,
\end{aligned}
\end{equation}
where the inequality ensures the homogeneity of the constraints.
In addition, all packets should be reliably decoded subject to each user's decoding error rate being lower than the threshold, $\epsilon^{\max}$.
Due to the monotonicity of the error rate in SINR and interference power in \eqref{eq:bler}, the error rate requirement can be formulated as a constraint on the interference power experienced by each user as
\begin{equation}\label{eq:const:max_interference_original}
\begin{aligned}
\sum_{k'\neq k}S_{k',k} \mathbf{1}_{\{z_k = z_{k'}\}}\leq \alpha ,\  \forall k ,
\end{aligned}
\end{equation}
where $\alpha$ indicates the maximum interference power of user $k$.
Note that the interference powers, indicated by $S_{k',k}$ $\forall k'\neq k$, have been normalized against the noise as \eqref{eq:measurable_path_loss} and \eqref{eq:gain_matrix}.

The system objective is to reduce the number $Z$ of slots in each period, i.e., the interval of users' transmissions while ensuring reliability, which is mathematically formulated as 
\begin{equation}\label{eq:prob:im_problem_original}
\begin{aligned}
(\mathbf{z}^*, Z^*) \triangleq\arg\min_{\mathbf{z}, Z}  Z  , \ \text{s.t. } 
\eqref{eq:const:Z_is_positive_integer},
\eqref{eq:const:z_k_in_Z_original},
\eqref{eq:const:same_bs_diff_slot_original},
\eqref{eq:const:max_interference_original}.
\end{aligned}
\end{equation}
$Z^*$ in \eqref{eq:prob:im_problem_original} denotes the minimum number of slots where feasible slot assignment decisions $\mathbf{z}^*$ exist such that all constraints are satisfied. 
The above problem is an integer programming problem that is NP-Hard \cite{subramanian2008minimum}.

\begin{rem}
Note that we assume that users' locations are static in the above formulation in this section, where user-BS associations and interference powers do not change over time. Under this assumption, the network states $\mathbf{S}$ and $\mathbf{Q}$ are fixed, and the slot assignment decisions $\mathbf{z}$ and $Z$ are decided at the initialization of the network and remain the same over time. 
Based on this static scenario, we will present our methods in Sections \ref{sec:graph_theory_analysis}, \ref{sec:sdp_relaxation_framework} and \ref{sec:mmw_interference_management}. Meanwhile, later in Section \ref{sec:online_architecture}, we will formulate and study these methods to constantly update the slot assignments $\mathbf{z}$ and $Z$ in dynamic networks where user-BS associations and interference power change over time, e.g., due to user mobility.
\end{rem}

\section{Sparse Interference Graph-Aided SDP}\label{sec:graph_theory_analysis}
\begin{figure}[!t]
\centering
\includegraphics[scale=0.85]{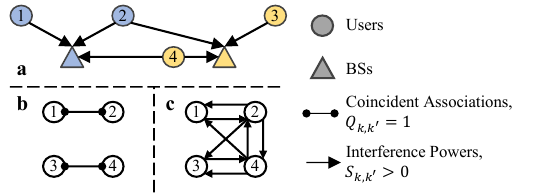}
\caption{A graphical illustration of WTSN: a) an interference network topology with 4 users and 2 BSs, and b) the association graph $\mathcal{G}^{\mathrm{asso}}$, and c) the interference-power graph $\mathcal{G}^{\mathrm{intp}}$.}
\label{fig:interference_graph_model}
\vspace{-0.4cm}
\end{figure}
This section first constructs graphs representing the interference and then shows the graphs' sparsity exploitation for interference management in the SIG-SDP framework.

\subsection{Definitions of Interference Graphs}
We define two graphs representing the interference in the network. 
Specifically, the first graph, $\mathcal{G}^{\mathrm{asso}} \triangleq (\mathcal{V}^{\mathrm{asso}},\mathcal{E}^{\mathrm{asso}})$, is a binary undirected graph indicating where a user pair is associated with the same BS, i.e., $\mathcal{V}^{\mathrm{asso}} \triangleq \{1,\dots,K\}$ and $\mathcal{E}^{\mathrm{asso}} \triangleq \{\{k,k'\}|Q_{k,k'}=1\}$.
The second graph is a directed weighted graph that collects the measured interference powers as $\mathcal{G}^{\mathrm{intp}} \triangleq (\mathcal{V}^{\mathrm{intp}},\mathcal{E}^{\mathrm{intp}})$, and its vertices and edges are $\mathcal{V}^{\mathrm{intp}} \triangleq \{1,\dots,K\}$ and $\mathcal{E}^{\mathrm{intp}} \triangleq \{(k,k')|S_{k,k'}>0, k\neq k' \}$,
where vertices are all users and an edge with a weight $S_{k,k'}$ connects two users if the interference power from user $k$ to the associated BS of user $k'$ is non-zero as \eqref{eq:measurable_path_loss}. 
The adjacency matrices of $\mathcal{G}^{\mathrm{asso}}$ and $\mathcal{G}^{\mathrm{intp}}$ correspond to $\mathbf{Q}$ in \eqref{eq:asso_matrix} and $\mathbf{S}$ in \eqref{eq:gain_matrix}, respectively. In other words, each non-zero element in $\mathbf{Q}$ and $\mathbf{S}$ corresponds to an edge in $\mathcal{G}^{\mathrm{asso}}$ and $\mathcal{G}^{\mathrm{intp}}$, respectively.
By defining the maximum number of neighbors of a user in $\mathcal{G}^{\mathrm{intp}}$, including in and out neighbors, as
\begin{equation}\label{eq:defi:graph_gain_max_neighbor}
\begin{aligned}
\Omega=\max_k |\{k'|(k,k') \ \text{or} \ (k',k) \in\mathcal{E}^{\mathrm{intp}}\}|,
\end{aligned}
\end{equation}
these two graphs' numbers of edges can be bounded based on the maximum number of neighbors in $\mathcal{G}^{\mathrm{intp}}$, $\Omega$, as
\begin{equation}
\begin{aligned}
|\mathcal{E}^{\mathrm{asso}}|\leq|\mathcal{E}^{\mathrm{intp}}|\leq K \Omega,
\end{aligned}
\end{equation}
which is because when two users $k$ and $k'$ are associated with the same BS, they interfere with each other's BS when sharing the same slot. Thus, they are both in and out neighbor of each other in $\mathcal{G}^{\mathrm{intp}}$.
Fig. \ref{fig:interference_graph_model}a shows a WTSN system with 4 users and 2 BSs, where users $1$ and $2$ are associated with the first BS, and users $3$ and $4$ are associated with the second one. 
Here, user $4$'s signal is also measurable at the first BS, and user $2$'s signal is measurable at the second BS. This means that user $4$ will interfere with users $1$ and $2$, and user $2$ will interfere with users $3$ and $4$ if they are assigned to the same slot.
Consequently, the association graph and the interference-power graph are constructed in Fig. \ref{fig:interference_graph_model}b and Fig. \ref{fig:interference_graph_model}c, respectively.

\subsection{SIG-SDP Framework for Interference Management}
The sparsity in the graphs $\mathcal{G}^{\mathrm{asso}}$ and $\mathcal{G}^{\mathrm{intp}}$ arises from the attenuation of the wireless signals. Specifically, user pairs that are far apart and associated with different BSs do not have an edge in $\mathcal{G}^{\mathrm{asso}}$, while pairs whose interference power falls below the measurable threshold do not have edges in $\mathcal{G}^{\mathrm{intp}}$. We establish the SIG-SDP framework that exploits the sparsity of these graphs for interference management as follows.

\subsubsection{Bounding Minimum Number of Slots in SDP}
We can provide the upper and the lower bounds on the minimum number $Z^*$ of slots based on the graphs $\mathcal{G}^{\mathrm{asso}}$ and $\mathcal{G}^{\mathrm{intp}}$ as
\begin{theorem}\label{theorem:min_slot_bounds}
The minimum number of slots, $Z^*$, in the interference management problem \eqref{eq:prob:im_problem_original} follows
\begin{equation}\label{eq:theorem:min_slot_bounds}
\begin{aligned}
1-\frac{\lambda_{\max}(\mathbf{Q})}{\lambda_{\min}(\mathbf{Q})}\leq \chi(\mathcal{G}^{\mathrm{asso}})\leq	Z^*\leq\chi(\mathcal{G}^{\mathrm{intp}}) \leq \Omega  + 1,
\end{aligned}
\end{equation}
where $\mathbf{Q}$ is the adjacency matrix of $\mathcal{G}^{\mathrm{asso}}$ and $\Omega$ is the maximum number of neighbors in $\mathcal{G}^{\mathrm{intp}}$.
$\chi(\mathcal{G}^{\mathrm{asso}})$ and $\chi(\mathcal{G}^{\mathrm{intp}})$ are chromatic numbers of $\mathcal{G}^{\mathrm{asso}}$ and $\mathcal{G}^{\mathrm{intp}}$, respectively.
\end{theorem} 
\begin{proof}
The proof uses chromatic numbers' bounds \cite{west2001introduction,hoffman2003eigenvalues} in $\mathcal{G}^{\mathrm{intp}}$ and $\mathcal{G}^{\mathrm{asso}}$ and is in the appendix.
\end{proof}
The bound in \eqref{eq:theorem:min_slot_bounds} shows a range of possible values of the minimum number of slots, $Z^*$, other than an unbounded range in \eqref{eq:const:Z_is_positive_integer}. 
This allows \eqref{eq:prob:im_problem_original} to be solved by iteratively checking whether there exists feasible $\mathbf{z}$ satisfying the remaining constraints \eqref{eq:const:z_k_in_Z_original}, \eqref{eq:const:same_bs_diff_slot_original}, and \eqref{eq:const:max_interference_original} for a given $Z$ within the bound in \eqref{eq:theorem:min_slot_bounds}, e.g., increasing $Z$ if no feasible $\mathbf{z}$ else decreasing $Z$ until $Z^*$ is found.
As a result, \eqref{eq:prob:im_problem_original} is reduced to a sequence of CSPs, where each CSP has optimization variables $\mathbf{z}$ as
\begin{equation}\label{eq:prob:csp_original}
\begin{aligned}
\mathrm{find}\ \mathbf{z},\ \text{s.t. }
\eqref{eq:const:z_k_in_Z_original},
\eqref{eq:const:same_bs_diff_slot_original},
\eqref{eq:const:max_interference_original},
\end{aligned}
\end{equation}
for given $Z$ in the bounds in \eqref{eq:theorem:min_slot_bounds}.
However, \eqref{eq:prob:csp_original} is still NP-hard to find the integer slot assignments $\mathbf{z}$ satisfying the constraints, while it has fewer optimization variables compared to \eqref{eq:prob:im_problem_original}, e.g., $Z$ is given.
We will explain the SDP relaxation that transforms the integer variables $\mathbf{z}$ to continuous ones and the search on the minimum number $Z^*$ of slots in Section \ref{sec:sdp_relaxation_framework}.

\subsubsection{Accelerating Computational Routines in SDP}
Given the SDP-relaxed problem, the solver needs to optimize the PSD matrix, $\mathbf{X}$, representing user pairs' assignment indicators $\mathbf{1}_{\{z_k = z_{k'}\}}$ $\forall k\neq k'$. Specifically, $\mathbf{X}$'s elements, $X_{i,j}$ $\forall k\neq k'$, indicate the likelihood that the user pair $(i, j)$ should be assigned to the same slot, i.e., how likely $\mathbf{1}_{\{z_k = z_{k'}\}} = 1$.  Denote the complexity of finding the optimal $\mathbf{X}$ in SDP as $\Gamma(K)$ for the given number of users, $K$. 
Directly applying widely-used SDP solvers, such as PDIP and ADMM, poses scalability challenges due to their high polynomial complexity, approximated as $\Gamma(K) \approx K^{\omega}$, where $\omega \gg 2$ \cite{jiang2020faster,o2016conic}. As a result, these solvers fail to return the optimized matrix in a timely manner as $K$ increases, making them unsuitable for scaling in large networks \cite{shi2015largeB}.

To address this complexity issue, we can leverage the sparsity of the constraint coefficients. Specifically, each constraint coefficient, $Q_{i,j}$ or $S_{i,j}$, in \eqref{eq:const:same_bs_diff_slot_original} and \eqref{eq:const:max_interference_original} corresponds to an edge $(i,j)$ in one of the graphs $\mathcal{G}^{\mathrm{asso}}$ or $\mathcal{G}^{\mathrm{intp}}$. 
The sparsity of these graphs implies that only a small subset of user pairs' assignments are constrained by \eqref{eq:const:same_bs_diff_slot_original} and \eqref{eq:const:max_interference_original}. Consequently, we can focus on optimizing only the elements $X_{i,j}$ that correspond to edges $(i,j)$ in either graph. This significantly reduces the complexity and enhances the scalability of the SDP solver.
To achieve this, we will design an accelerated SDP solver using the MMW algorithm, as explained in Section \ref{sec:mmw_interference_management}.

\section{Relaxation and Sparsity-Aware Binary Search in SIG-SDP}\label{sec:sdp_relaxation_framework}
This section first presents the SDP relaxation of the slot assignment problem in \eqref{eq:prob:csp_original} for the given number of slots, $Z$. We then design a binary search iterating $Z$ to find the minimum number of slots within the bounds derived in Theorem \ref{theorem:min_slot_bounds}.

\subsection{SDP Relaxation of Slot Assignment Task}\label{subsec:sdp_relaxation_im}
\begin{figure}[!t]
\centering
\includegraphics[scale=0.725]{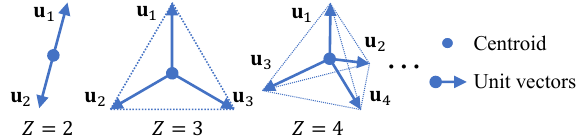}
\caption{Unit radial vectors in  regular $(Z-1)$-simplexes, with $Z=2,3,4,\dots$, representing all possible slot choices for given $Z$.}
\label{fig:simplexes}
\vspace{-0.4cm}
\end{figure}

First, for the given $Z$, we represent all possible slot choices $\{ 1, \dots, Z \}$ as vectors $\{ \mathbf{u}_1, \dots, \mathbf{u}_Z \}$, where each slot $z$ corresponds to a vector $\mathbf{u}_z$, $z=1,\dots,Z$. 
As shown in Fig. \ref{fig:simplexes}, each $\mathbf{u}_z$ is the unit radial vector from the centroid of the regular $(Z-1)$-simplex to the $z$-th vertex (for simplicity, we assume they have been normalized, i.e., $\|\mathbf{u}_z\|=1$ $\forall z$). 
The inner product of any two vectors in $\{ \mathbf{u}_1, \dots, \mathbf{u}_Z \}$ \cite{frieze1997improved} is 
\begin{equation}\label{eq:inner_product_of_uu_cases}
\begin{aligned}
\langle \mathbf{u}_z,\mathbf{u}_{z'}\rangle = 
\begin{cases}
 1 &, \ \text{if } z = z' , \\
-\frac{1}{Z-1} &,\ \text{otherwise} .
\end{cases}
\end{aligned}
\end{equation}
For user $k$, deciding which slot in $\{ 1, \dots, Z \}$ is assigned as $z_k$ is equivalent to deciding which vector in $\{ \mathbf{u}_1, \dots, \mathbf{u}_Z \}$ is assigned to user $k$, i.e., the constraint in \eqref{eq:const:z_k_in_Z_original} is rewritten as
\begin{equation}\label{eq:vector_representation_of_z}
\begin{aligned}
\mathbf{v}_{k} \in \{ \mathbf{u}_1, \dots, \mathbf{u}_Z \}, \forall k  ,
\end{aligned}
\end{equation}
where $\mathbf{v}_{k}$ is the vector choice of user $k$. 
Using the property of vector choices in \eqref{eq:inner_product_of_uu_cases}, we can express the indicator function on whether two users $k$ and $k'$ are in the same slot or not as
\begin{equation}\label{eq:indicator_z_z_equal_X}
\begin{aligned}
\mathbf{1}_{\{z_k=z_{k'}\}}= \mathbf{1}_{\{\mathbf{v}_k=\mathbf{v}_{k'}\}} &= \frac{1}{Z}( 1 + (Z-1) \langle \mathbf{v}_{k},\mathbf{v}_{k'}\rangle) ,
\end{aligned}
\end{equation}
which is a function on the inner product between vector choices, $\langle \mathbf{v}_{k},\mathbf{v}_{k'}\rangle$.
We use a $K\times K$ matrix $\mathbf{X}$ to represent inner products between all user pairs' vector choices (or the slot assignment decisions) as
\begin{equation}\label{eq:const:vv_to_X}
\begin{aligned}
\mathbf{X} \triangleq [\mathbf{v}_1,\dots,\mathbf{v}_K]^{\rm T}[\mathbf{v}_1,\dots,\mathbf{v}_K] ,
\end{aligned}
\end{equation}
where $X_{k,k'}$ represents the inner product between $\mathbf{v}_k$ and $\mathbf{v}_{k'}$, $\langle \mathbf{v}_k,\mathbf{v}_{k'}\rangle$ $\forall k,k'$, and $X_{k,k}=1$ $\forall k $ (as all vector choices are normalized). $\mathbf{X}$ is a symmetric positive semidefinite matrix. 
Elements of $\mathbf{X}$ replace the inner product and rewrite the indication function in \eqref{eq:indicator_z_z_equal_X} and the constraint in \eqref{eq:const:same_bs_diff_slot_original} as
\begin{equation}\label{eq:const:same_bs_diff_slot_X}
\begin{aligned}
\frac{1}{Z}\big( 1 + (Z-1) X_{k,k'} \big) \leq 0 , \ \forall Q_{k,k'} =1 .
\end{aligned}
\end{equation}
Also, the interference power constraint \eqref{eq:const:max_interference_original} is recast as
\begin{equation}\label{eq:const:max_interference_X}
\begin{aligned}
\sum_{k'\neq k}S_{k',k} \cdot \frac{1}{Z}\big( 1 + (Z-1) X_{k',k}\big)\leq\alpha ,\ \forall k .
\end{aligned}
\end{equation}

The relaxation of \eqref{eq:prob:csp_original} is made by removing the constraints \eqref{eq:vector_representation_of_z} and \eqref{eq:const:vv_to_X} on $\mathbf{X}$ being the inner products of the simplex's unit radial vectors and by allowing $\mathbf{X}$ to be any PSD matrix with all-$1$ on its diagonal, e.g.,
\begin{equation}\label{eq:const:X_psd_diag1}
\begin{aligned}
\mathbf{X} \succeq 0, \ X_{k,k} = 1 ,\ \forall k.
\end{aligned}
\end{equation}
In other words, the CSP in \eqref{eq:prob:csp_original} is relaxed to
\begin{equation}\label{eq:prob:csp_X_relaxed}
\begin{aligned}
\mathrm{find}\ \mathbf{X} , \ \text{s.t. }
\eqref{eq:const:same_bs_diff_slot_X},
\eqref{eq:const:max_interference_X},\eqref{eq:const:X_psd_diag1} .
\end{aligned}
\end{equation}
Here, the integer slot (or vector) choices in \eqref{eq:prob:csp_original} are relaxed to the continuous variables in a PSD matrix in \eqref{eq:prob:csp_X_relaxed}, i.e., the integer CSP \eqref{eq:prob:csp_original} is relaxed into the SDP CSP \eqref{eq:prob:csp_X_relaxed}.

\subsection{Sparsity-Aware Binary Search Using Derived Bounds}\label{subsec:binary_search_flow}
The binary search initializes the range on the minimum number of slots with the lower and the upper bounds as $Z^\mathrm{a}=\lceil 1-\lambda_{\max}(\mathbf{Q})/\lambda_{\min}(\mathbf{Q})\rceil$ and  $Z^\mathrm{b}=\Omega +1$, respectively, as stated in Theorem~\ref{theorem:min_slot_bounds}. 
In each binary search iteration, the searched $Z$ is set to the average of the upper and the lower bounds, i.e., $Z=\lfloor(Z^{\mathrm{a}}+Z^{\mathrm{b}})/2\rfloor$ and the relaxed CSP in \eqref{eq:prob:csp_X_relaxed} is solved with the given $Z$ using the SDP solver as $\mathrm{sdpslv}(\cdot)$. The solver takes the users' interference powers $\mathbf{S}$, associations $\mathbf{Q}$ and the number of slots $Z$ as its input and returns a gram matrix form of the optimal PSD matrix of \eqref{eq:prob:csp_X_relaxed}, $\mathbf{X}'$, i.e.,
\begin{equation}\label{eq:sdp_solver_procedure}
\begin{aligned}
\gram(\mathbf{X}') \triangleq   [\mathbf{v}'_1,\dots,\mathbf{v}'_K]^{\rm T} =  \mathrm{sdpslv}(\mathbf{S},\mathbf{Q},Z),
\end{aligned}
\end{equation}
where $\mathbf{X}'=[\mathbf{v}'_1,\dots,\mathbf{v}'_K]^{\rm T}  [\mathbf{v}'_1,\dots,\mathbf{v}'_K]$. 
Note that $\mathbf{v}'_k$ are $D$-dimensional vectors, $\forall k$, where $D$ depends on the solver's implementation and is a constant at a value of $\mathcal{O}(Z)$, which we shall discuss in detail later in Section \ref{sec:mmw_interference_management}.

Here, $\gram(\mathbf{X}')$ are relaxed vector choices, of which the inner products indicate how closer any two user's vector choices are and how likely two users should be in the same slot \cite{goemans1995improved,frieze1997improved}. 
We round $\gram(\mathbf{X}')$ to the integer slot assignments in each binary search iteration.
Let $\mathcal{V}^{z}$ be initialized as an empty set for each iteration to collect users in slot $z$, $z=1,\dots,Z$.
We generate $Z$ random unit vector, $\delta^1,\dots,\delta^Z$, with the same dimension as $\mathbf{v}'_k$, $\forall k$, and sort the slot indices, $z=1,\dots,Z$, based on the inner product between $\delta^z$ and $\mathbf{v}'_k$ in descending order for each user $k=1,\dots,K$ as
\begin{equation}\label{eq:rounding_order}
\begin{aligned}
z^{k}_1, z^{k}_2,\dots, z^{k}_Z ,\ \text{where}\ \langle \delta^{z^{k}_i}, \mathbf{v}'_{k} \rangle\geq \langle \delta^{z^{k}_j}, \mathbf{v}'_{k}  \rangle ,\ i<j.
\end{aligned}
\end{equation}
Then, $\mathcal{V}^{z}$ is updated for each slot $z=z^{k}_1, z^{k}_2,\dots, z^{k}_Z$ in the above list for the given user $k$ by testing the satisfaction of integer constraints \eqref{eq:const:same_bs_diff_slot_original}\eqref{eq:const:max_interference_original} (before relaxation) for all neighboring users' and user $k$ itself as
\begin{equation}\label{eq:rounding_check_per_user}
\begin{aligned}
\mathcal{V}^{z} \leftarrow 
\begin{cases}
\mathcal{V}^{z}\cup\{k\},
&\forall k' \in\mathcal{V}^{z}, Q_{k,k'} = 0,\\
&\text{and } \forall k' \in\mathcal{V}^{z}\cup\{k\}, \\
&\sum_{k''\in (\mathcal{V}^{z}\cup\{k\})} S_{k'',k'}\leq \alpha, \\
\mathcal{V}^{z} ,\ &\text{otherwise}, \ 
\end{cases}
\end{aligned}
\end{equation}
where if all constraints are satisfied then we add $k$ in $\mathcal{V}^{z}$ and move to the next user $k+1$; or otherwise $\mathcal{V}^{z}$ remains the same and we test the next slot in \eqref{eq:rounding_order} for user $k$.
Here, if two users' vector choices are close, they will be close to the same random vector, and the corresponding slot choices are ranked at the front of the list in \eqref{eq:rounding_order}. Consequently, these two users are more likely to be allocated in the same slot \cite{frieze1997improved}.

After all users' slot choices are tested, if all users are assigned to a slot, i.e., $|\cup_z \mathcal{V}^{z}| = K$, then $Z$ is feasible, and we set the upper bound $Z^\mathrm{b}$ to $Z$ to search for a smaller slot number. Otherwise, when some users are not assigned, we then adjust the lower bound $Z^\mathrm{a}$ to $Z + 1$ to search for a larger slot number. 
The binary search repeats the above process in the next iteration for  $Z=\lfloor(Z^{\mathrm{a}}+Z^{\mathrm{b}})/2\rfloor$ and stops until $Z^\mathrm{a}$ and $Z^\mathrm{b}$ are equal, which implies that the minimum number of slots is found as $Z^\mathrm{a}$ (or $Z^\mathrm{b}$). 
The final slot assignment is configured as the one in the last iteration as 
\begin{equation}\label{eq:slot_set_of_rounded_users}
\begin{aligned}
z_k = z,\ \forall k \in \mathcal{V}^{z} ,\ \forall z = 1,\dots,Z  .
\end{aligned}
\end{equation}
The binary search of SIG-SDP for the aforementioned WTSN interference management is summarized in Algorithm~\ref{alg:im-sdp_framework}.

\begin{algorithm}[!t]
\caption{Sparsity-Aware Binary Search of SIG-SDP}
\label{alg:im-sdp_framework}
\begin{algorithmic}[1]
\PROCEDURE{$\mathrm{SdpBinarySearch}$}{$\mathbf{S},\mathbf{Q}$}
\STATE Initialize $Z^{\mathrm{a}}=\lceil 1-\lambda_{\max}(\mathbf{Q})/\lambda_{\min}(\mathbf{Q})\rceil$.
\STATE Initialize $Z^{\mathrm{b}}=\Omega+1$.

\FOR{$m=1,2,\dots$}
    \STATE $Z\leftarrow \lfloor(Z^{\mathrm{a}}+Z^{\mathrm{b}})/2\rfloor$.
    \STATE Call $\gram(\mathbf{X}') = \mathrm{sdpslv}(\mathbf{S},\mathbf{Q},Z)$ to solve \eqref{eq:prob:csp_X_relaxed}.\label{alg:line:sdp_solver}
    \STATE Round $\gram(\mathbf{X}')$ to $\mathcal{V}^{1},\dots,\mathcal{V}^{Z}$ as \eqref{eq:rounding_order} and \eqref{eq:rounding_check_per_user}.
    \STATE \textbf{if} $|\cup_z \mathcal{V}^{z}| = K$ \textbf{then} Set $Z^{\mathrm{b}}=Z$.
    \STATE \textbf{else} Set $Z^{\mathrm{a}}=Z+1$.
    \STATE \textbf{if} $Z^{\mathrm{a}}=Z^{\mathrm{b}}$ \textbf{then} Set $z_k$ as \eqref{eq:slot_set_of_rounded_users}, and \textbf{break}.
\ENDFOR
\STATE \textbf{return} $\mathbf{X}'$, $Z$ and $\mathbf{z}$.
\ENDPROCEDURE\label{alg:line:rounding_end}
\end{algorithmic}
\end{algorithm}

\subsection{Complexity of Binary Search of SIG-SDP in Algorithm \ref{alg:im-sdp_framework}}\label{subsec:sdp_framework_complexity}
The initialization of the slot bounds requires the maximum and minimum eigenvalues of $\mathbf{Q}$, of which the computation is approximately at the complexity of the number of non-zero elements in $\mathbf{Q}$, i.e., $\mathcal{O}(\Omega K)$, using iterative eigenvalue algorithms \cite{lehoucq1998arpack}.
The rounding is repeated for $K$ users, and rounding for each user requires computing the inner products between the random vector and the user's vector choices (with $\mathcal{O}(ZD)\approx\mathcal{O}(\Omega^2)$ complexity), sorting of the inner products (with $\mathcal{O}(Z\log Z)\approx \mathcal{O}(\Omega\log \Omega)$ complexity) and checking the interference constraints (with $\mathcal{O}(Z\Omega)\approx\mathcal{O}(\Omega^2)$ complexity as only neighboring users' interference requires processing). Here, $Z$ and $D$ is approximately at $\mathcal{O}(\Omega)$.
Moreover, the number of iterations in the binary search is upper bounded by $\mathcal{O}(\log\Omega)$, considering the search range.
The major complexity of the framework is at the SDP solver, and we will reduce it by exploiting the interference graphs' sparsity in the next.



\section{Sparsity-Accelerated MMW Solver in SIG-SDP}\label{sec:mmw_interference_management}
In this section, we design the MMW-based SDP solver to solve the SDP CSP in \eqref{eq:prob:csp_X_relaxed}. We first format the SDP CSP in \eqref{eq:prob:csp_X_relaxed} in a canonical form and then provide an overview of how the MMW solves the formatted problem, followed by the detailed implementation exploiting the graph sparsity and the complexity analysis.
\subsection{Formatting the Canonical Form of the SDP CSP}\label{subsec:canonical_form_sdp}
The SDP CSP \eqref{eq:prob:csp_X_relaxed} is formatted in the canonical form as
\begin{defi}[Canonical Form of SDP CSPs]
Let $\mathcal{X}\subseteq \mathbb{R}^{K\times K} $, $\mathcal{X}\triangleq\{\mathbf{X}|\mathbf{X}\succeq 0, \Tr(\mathbf{X})=K\}$. The primal SDP CSP is defined as
\begin{equation}\label{eq:definition:sdp_canonical_form_primal}
\begin{aligned}
\mathrm{find}\ \mathbf{X}\in \mathcal{X},\  \text{s.t. } \mathbf{A}^{(c)}\bullet\mathbf{X}\leq0, c=1,\dots,C,
\end{aligned}
\end{equation}
where $C$ is the number of constraints and $\mathbf{A}^{(c)}$ is the  coefficient matrix of the $c$-th constraint. Here, all coefficient matrices are symmetric and normalized in spectral norm as
\begin{equation}
\begin{aligned}
\mathbf{A}^{(c)}=(\mathbf{A}^{(c)})^{\rm T}  ,\ \|\mathbf{A}^{(c)}\|=1 ,\ \forall c  . 
\end{aligned}
\end{equation}
Let $\mathcal{Y}\subseteq \mathbb{R}^{C\times 1}$ and $\mathcal{Y}\triangleq\{\mathbf{y}|\mathbf{y}\geq 0,\ |\mathbf{y}|=1  \}$. Then, the corresponding dual problem of \eqref{eq:definition:sdp_canonical_form_primal} is
\begin{equation}\label{eq:definition:sdp_canonical_form_dual}
\begin{aligned}
\mathrm{find}\ \mathbf{y}\in \mathcal{Y},\  \text{s.t. } \sum^{C}_{c=1} y_c \cdot \mathbf{A}^{(c)} \succeq 0 .
\end{aligned}
\end{equation}
\end{defi}
The detailed definition of the coefficient matrices for the canonical form of problem \eqref{eq:prob:csp_X_relaxed} is explained in the appendix. 
Note that the total number of constraints, $C$, in the canonical form of problem \eqref{eq:prob:csp_X_relaxed} follows
\begin{equation}\label{eq:linearity_in_K_and_C}
\begin{aligned}
C \triangleq 2K+|\mathcal{E}^{\mathrm{asso}}| \leq 2K+K\Omega \approx \mathcal{O}(K\Omega),
\end{aligned}
\end{equation}
where $\Omega$ is the upper bound of a user's maximum number of neighbors in interference graphs, as defined in Section \ref{sec:graph_theory_analysis}. The duality in the above canonical form can be stated as
\begin{theorem}\label{theorem:duality_of_the_canonical_form}
The duality of the primal and dual SDP CSP problems can be described as follows. Define an indicator,
\begin{equation}\label{eq:definition:sdp_canonical_form_feasibility_indicator}
\begin{aligned}
\mathcal{J}\triangleq\min_{\mathbf{X}\in \mathcal{X}}\max_{\mathbf{y}\in \mathcal{Y}} \sum_c y_c \cdot \mathbf{A}^{(c)} \bullet \mathbf{X} ;
\end{aligned}
\end{equation}
Define the duality gap between primal and dual variables as
\begin{equation}\label{eq:definition:sdp_canonical_form_duality_gap}
\begin{aligned}
\mathrm{gap}(\mathbf{y},\mathbf{X}) \triangleq \max_c \mathbf{A}^{(c)} \bullet \mathbf{X}- \lambda_{\min}(\sum_c y_c \cdot \mathbf{A}^{(c)})K .
\end{aligned}
\end{equation}
Then, $\mathrm{gap}(\mathbf{y},\mathbf{X})\geq 0 $ for all $\mathbf{y}\in\mathcal{Y}$ and $\mathbf{X}\in\mathcal{X}$.
Furthermore, there exist $\mathbf{y}^*\in\mathcal{Y}$ and $\mathbf{X}^*\in\mathcal{X}$ such that $\mathrm{gap}(\mathbf{y}^*,\mathbf{X}^*)=0$, i.e., the strong duality holds between  \eqref{eq:definition:sdp_canonical_form_primal} and  \eqref{eq:definition:sdp_canonical_form_dual}. When the duality gap is $0$,  $\lambda_{\min}(\sum_c y^*_c \cdot \mathbf{A}^{(c)})K =\max_c \mathbf{A}^{(c)} \bullet \mathbf{X}^*=\mathcal{J}$.
\end{theorem}
\begin{proof}
    The proof is in the appendix.
\end{proof}

\subsection{Overview of MMW for Interference Management}
\begin{figure}[!t]
\centering
\includegraphics[scale=0.675]{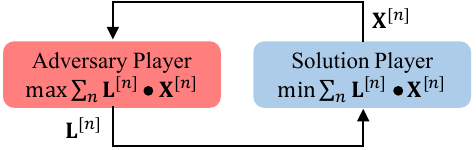}
\vspace{-0.1cm}
\caption{Illustration of the matrix multiplicative weights algorithm as a turn-based zero-sum game of two players for turns $n=1,\dots,N$.}
\label{fig:mmw_two_players}
\vspace{-0.4cm}
\end{figure}
A MMW algorithm can be viewed as a turn-based zero-sum game of two players \cite{arora2007combinatorial}, as shown in Fig. \ref{fig:mmw_two_players}. Specifically, in the $n$-th turn, $n=1,\dots,N$, an adversary player first generates a symmetric loss matrix $\mathbf{L}^{[n]}$ that has a bounded spectral norm (without loss of generality, we assume that the upper bound on the loss matrix's norm is $1$) as
\begin{equation}\label{eq:mmw:loss_condition}
\begin{aligned}
\mathbf{L}^{[n]} = (\mathbf{L}^{[n]})^{\rm T} , \ \|\mathbf{L}^{[n]}\| \leq 1  ,\ \forall n  . 
\end{aligned}
\end{equation}
On the other hand, the solution player generates a PSD solution matrix with its trace equal to $K$ as
\begin{equation}\label{eq:mmw:solution_condition}
\begin{aligned}
\mathbf{X}^{[n]}\succeq 0 , \ \Tr (\mathbf{X}^{[n]})= K , \ \forall n  .
\end{aligned}
\end{equation}

The adversary or solution player aims to maximize or minimize the loss (i.e., the solution player's loss is the adversary player's reward), respectively, where the loss is $\sum_{n=1}^{N}\mathbf{L}^{[n]}\bullet\mathbf{X}^{[n]}$ in all played turns. This structure can be used to solve the interference management SDP CSP by setting the loss as the constraint violations of given PSD matrices, i.e.,  minimizing/maximizing the loss is equivalent to minimizing/maximizing the constraint violations.
For example, in each turn, the adversary player generates the loss matrix $\mathbf{L}^{[n]}$ as a weighted combination of the constraint coefficient matrices as $\sum_c y^{[n]}_c \mathbf{A}^{(c)}$, where $y^{[n]}_c$ is the weight of the $c$-th constraint in the $n$-th turn. To maximize the loss $\sum_n\mathbf{L}^{[n]}\bullet\mathbf{X}^{[n]} = \sum_n\sum_c y^{[n]}_c \mathbf{A}^{(c)}\bullet\mathbf{X}^{[n]}$, the adversary player sets $y^{[n]}_c$ to a larger value if the $c$-th constraint has a higher violation for the PSD matrices generated by the solution player. 
On the other hand, the solution player uses matrix exponential to approximate the PSD matrix $\mathbf{X}^{[n]}$ that minimizes total constraint violations for given weights of constraints, $\sum_n\sum_c y^{[n]}_c \mathbf{A}^{(c)}\bullet\mathbf{X}^{[n]}$.
In this structure, $\mathbf{X}^{[n]}$ and $[y^{[n]}_1,\dots,y^{[n]}_C]$ correspond to the iterated primal and dual variables, respectively, in the canonical SDP CSP.
The detailed implementation of the MMW is explained in the next.
\subsection{Implementation of MMW}
First, the solver formats the canonical form of \eqref{eq:prob:csp_X_relaxed} according to network states $\mathbf{S}$ and $\mathbf{Q}$ and the given slot number $Z$.
Then, we initialize a constant step size $\eta$, $\eta< 1$, to control the update rate of variables and a constant $D$, $D>1$, to control the accuracy when computing matrix exponential. 
In the $n$-th turn, $n=1,\dots,N$, the adversary player configures the weights of constraints (or the dual variables) using the hedge rule \cite{freund1997decision} as
\begin{equation}\label{eq:mmw:vector_exponential}
\begin{aligned}
&\mathbf{y}^{[n]}\triangleq [y_1^{[n]},\dots,y_C^{[n]}] = \\
&\mathrm{Softmax}([\eta \sum_{n'=1}^{n}\mathbf{A}^{(1)}\bullet\mathbf{X}^{[n']},\dots, \eta \sum_{n'=1}^{n}\mathbf{A}^{(C)}\bullet\mathbf{X}^{[n']}]),
\end{aligned}
\end{equation}
where $\mathrm{Softmax}(\cdot)$ is the softmax function that outputs $\mathbf{y}^{[n]}\in \mathcal{Y}$. Here, $\mathbf{X}^{[1]} = \mathbb{I}^{K}$ as the PSD matrix in the first turn.
In \eqref{eq:mmw:vector_exponential}, $y_c^{[n]}$ is larger if the $c$-th constraint's accumulated violation is relatively higher than other constraints.
The loss matrix in each iteration is generated as
\begin{equation}\label{eq:mmw:loss_matrix}
\begin{aligned}
\mathbf{L}^{[n]} = \sum_{c=1}^{C} y_c^{[n]} \cdot \mathbf{A}^{(c)}, \ \forall n  , 
\end{aligned}
\end{equation}
as mentioned before. Note that the constraint coefficient matrices $\mathbf{A}^{(c)}$ have been normalized with a spectrum norm $1$ in the canonical form, $\forall c$. Thus, the loss matrix in the above has a spectrum norm less than $1$ due to the triangle inequality (as $\sum_{c} y_c^{[n]}=1$). Also, the loss matrix is symmetric because the constraint coefficient matrices are symmetric. Therefore, the condition in \eqref{eq:mmw:loss_condition} is satisfied.

The solution player in MMW methods \cite{arora2007combinatorial} specifies a matrix exponential on the sum of played loss matrices to generate the solution matrix in each turn as
\begin{equation}\label{eq:mmw:matrix_exponential}
\begin{aligned}
\mathbf{X}^{[n+1]} = \frac{K\exp(-\eta \sum_{n'=1}^{n}\mathbf{L}^{[n']})}{\Tr\big(\exp(-\eta \sum_{n'=1}^{n}\mathbf{L}^{[n']})\big)}, \ \forall n  ,
\end{aligned}
\end{equation}
which satisfies the condition in \eqref{eq:mmw:solution_condition} and is a validate solution in $\mathcal{X}$.
We use randomized sketching \cite{carmon2019rank} that approximates the matrix exponential in \eqref{eq:mmw:matrix_exponential}. In detail, let 
\begin{equation}\label{eq:mmw:expm_approximation}
\begin{aligned}
\big[\mathbf{v}^{[n]}_1,\dots,\mathbf{v}^{[n]}_K \big]^{\rm T} \triangleq \exp(-\frac{\eta}{2} \sum_{n'=1}^{n}\mathbf{L}^{[n']}) \mathbf{R}^{[n]}  ,
\end{aligned}
\end{equation}
where $\mathbf{R}^{[n]}$ is a random $K\times D$ matrix with each row normalized and the larger $D$ is, the closer the approximation of the matrix exponential will be.
Then, the matrix exponential in the numerator of \eqref{eq:mmw:matrix_exponential} can be approximated as \cite{carmon2019rank}
\begin{equation}\label{eq:mmw:inner_product_of_matrix_exponential}
\begin{aligned}
&\exp(-\eta \sum_{n'=1}^{n}\mathbf{L}^{[n']}) \stackrel{\text{(a)}}{=} \exp(-\frac{\eta}{2} \sum_{n'=1}^{n}\mathbf{L}^{[n']})  \exp(-\frac{\eta}{2}  \sum_{n'=1}^{n}\mathbf{L}^{[n']}) \\
&=\exp(-\frac{\eta}{2}  \sum_{n'=1}^{n}\mathbf{L}^{[n']}) \mathbb{I}^{K} \exp(-\frac{\eta}{2}  \sum_{n'=1}^{n}\mathbf{L}^{[n']}) \\
&\stackrel{\text{(b)}}{\approx}\exp(-\frac{\eta}{2}  \sum_{n'=1}^{n}\mathbf{L}^{[n']}) \mathbf{R}^{[n]} (\mathbf{R}^{[n]})^{\rm T}  \exp(-\frac{\eta}{2}  \sum_{n'=1}^{n}\mathbf{L}^{[n']}) \\
&=\big[\mathbf{v}^{[n]}_1,\dots,\mathbf{v}^{[n]}_K \big]^{\rm T} \big[\mathbf{v}^{[n]}_1,\dots,\mathbf{v}^{[n]}_K \big],
\end{aligned}
\end{equation}
where (a) is because all loss matrices are symmetric \cite{arora2007combinatorial} and (b) is because the expected value of $\mathbf{R}^{[n]} (\mathbf{R}^{[n]})^{\rm T} $ is the identity matrix.
The matrix trace in the denominator of \eqref{eq:mmw:matrix_exponential} can be approximated as
\begin{equation}\label{eq:mmw:trace}
\begin{aligned}
\Tr\big(\exp(-\eta \sum_{n'=1}^{n}\mathbf{L}^{[n']})\big) \approx \sum_{k=1}^{K}\|\mathbf{v}^{[n]}_k\|.
\end{aligned}
\end{equation}
Thus, the PSD matrix in \eqref{eq:mmw:matrix_exponential} in the $n+1$-th turn is computed based on \eqref{eq:mmw:expm_approximation}\eqref{eq:mmw:inner_product_of_matrix_exponential}\eqref{eq:mmw:trace} as
\begin{equation}\label{eq:mmw:matrix_exponential_vector_representaion}
\begin{aligned}
\mathbf{X}^{[n+1]} \approx \frac{K}{\sum_{k=1}^{K}\|\mathbf{v}^{[n]}_k\|}\big[\mathbf{v}^{[n]}_1,\dots,\mathbf{v}^{[n]}_K \big]^{\rm T} \big[\mathbf{v}^{[n]}_1,\dots,\mathbf{v}^{[n]}_K \big].
\end{aligned}
\end{equation}
The above processes are repeated for $N$ times. We average the iterated primal and dual variables as
\begin{equation}
\begin{aligned}
\Bar{\textbf{X}} \triangleq \frac{1}{N}\sum_{n=1}^{N}\mathbf{X}^{[n]}, \ \Bar{\textbf{y}} \triangleq \frac{1}{N}\sum_{n=1}^{N}\mathbf{y}^{[n]}.
\end{aligned}
\end{equation}
We use truncated singular value decomposition (TSVD) \cite{strang2006linear} to approximate $\Bar{\textbf{X}}$ using its $D$ most significant components as
\begin{equation}\label{eq:mmw:tsvd_x_bar}
\begin{aligned}
\mathbf{U}  \mathbf{\Sigma}  \mathbf{U}^{\rm T} = \mathbf{tsvd}(\Bar{\textbf{X}},D),
\end{aligned}
\end{equation}
where $\mathbf{U}$ and $\mathbf{\Sigma}$ is a $K\times D $ and $D\times D$ matrix, respectively, and $\mathbf{\Sigma}$ is a diagonal matrix with the largest $D$ eigenvalues of $\Bar{\textbf{X}}$ on its diagonal and $\mathbf{U}$ collects the corresponding eigenvectors \cite{strang2006linear}. The gram form of $\Bar{\textbf{X}}$ is approximated as
\begin{equation}\label{eq:mmw:gram_x_bar}
\begin{aligned}
\gram(\Bar{\textbf{X}}) \approx  \mathbf{U} (\mathbf{\Sigma})^{\frac{1}{2}} .
\end{aligned}
\end{equation}
Finally, $\mathbf{U}(\mathbf{\Sigma})^{\frac{1}{2}}$ is returned from the solver to the SDP framework. The MMW for interference management is listed in Algorithm \ref{alg:mmw-im} and is interfaced with the binary search of SIG-SDP in line \ref{alg:line:sdp_solver} of Algorithm \ref{alg:im-sdp_framework} as $\mathrm{mmwsdpslv}(\mathbf{S},\mathbf{Q},Z)$.

\begin{algorithm}[!t]
\caption{MMW for Interference Management}\label{alg:mmw-im}
\begin{algorithmic}[1]
\PROCEDURE{$\mathrm{mmwsdpslv}$}{$\mathbf{S},\mathbf{Q},Z$}
\STATE Format the canonical form \eqref{eq:prob:csp_X_relaxed} as Section \ref{subsec:canonical_form_sdp}.
\STATE Initialize $\eta$, $D$.
\FOR{$n$ = $1,\dots,N$}
    \STATE Compute $\mathbf{y}^{[n]}$ as \eqref{eq:mmw:vector_exponential}.
    \STATE Compute $\mathbf{L}^{[n]}$ as \eqref{eq:mmw:loss_matrix}.
    \STATE Compute $\mathbf{X}^{[n+1]}$ in \eqref{eq:mmw:matrix_exponential} as \eqref{eq:mmw:expm_approximation}-\eqref{eq:mmw:matrix_exponential_vector_representaion}.
\ENDFOR
\STATE Compute the TSVD of $\Bar{\textbf{X}}$ as \eqref{eq:mmw:tsvd_x_bar}.
\STATE \textbf{return} $\gram(\Bar{\textbf{X}})$ as \eqref{eq:mmw:gram_x_bar}.
\ENDPROCEDURE
\end{algorithmic}
\end{algorithm}

\subsection{Convergence of MMW}
Assuming randomized sketching of the matrix exponential \eqref{eq:mmw:matrix_exponential} in \eqref{eq:mmw:expm_approximation}-\eqref{eq:mmw:matrix_exponential_vector_representaion} is exact, the convergence of the primal and dual variables iterated in MMW for interference management in Algorithm \ref{alg:mmw-im} can be stated as follows.
\begin{theorem}\label{theorem:mmw_convergence}
Let the number of iterations in the MMW, $N= \frac{1}{\eta^2}(\ln K + \ln C)$. Then, the duality gap for the average of iterated primal and dual variables, $\mathrm{gap}(\Bar{\mathbf{y}},\Bar{\mathbf{X}})$, in MMW of Algorithm \ref{alg:mmw-im} is less than $\mathcal{O}(\eta K)$.
\end{theorem}
\begin{proof}
    The proof is in the appendix.
\end{proof}

Theorem \ref{theorem:mmw_convergence} shows the duality gap converges to $0$ as we set $\eta$ to a smaller value. Based on Theorem \ref{theorem:duality_of_the_canonical_form}, Theorem \ref{theorem:mmw_convergence} implies the maximum constraint violation converges to the minimum value $\mathcal{J}$. 
Meanwhile, note that Theorem \ref{theorem:mmw_convergence} is based on the assumption that the randomized sketching in \eqref{eq:mmw:matrix_exponential_vector_representaion} approximates the matrix exponential exactly. In practice, the approximation error of the randomized sketching will increase the duality gap in iterations \cite{carmon2019rank}. 
To reduce this error, we can set $D$ larger so that $\mathbf{R}^{[n]}(\mathbf{R}^{[n]})^{\rm T}$ is closer to the identity matrix and the approximation in \eqref{eq:mmw:expm_approximation}-\eqref{eq:mmw:matrix_exponential_vector_representaion} is more accurate.
Nevertheless, the error only scales the gap in Theorem \ref{theorem:mmw_convergence} up to a constant factor (greater than $1$), as shown in \cite{carmon2019rank}, which remains the same as $\mathcal{O}(\eta K)$ in big-$\mathcal{O}$ notation.
On the other hand, for a given $Z$, the PSD matrix (before relaxation) is the inner product of vector choices from the $(Z-1)$-simplex, as discussed in Section \ref{subsec:sdp_relaxation_im}, which means the matrix has a low rank at most $Z-1$.
Ideally, we want the solver to find the optimal PSD matrix that has a rank lower than $Z-1$ or a rank around $Z-1$ in the relaxed problem. 
Note that the PSD matrix optimized by the MMW has a rank $D$ due to the random sketching in \eqref{eq:mmw:expm_approximation}-\eqref{eq:mmw:matrix_exponential_vector_representaion}.
Based on the above facts, we set $D \approx \mathcal{O} (Z)$, e.g., $\lceil\beta (Z-1)\rceil$, where $\beta (Z-1)\ll K$ and $\beta \geq 1$ for a constant $\beta$, leading to low-rank PSD matrices while keeping a low random approximation error that can be ignored in practice. 

\subsection{MMW's Routines and Complexity Exploiting Sparsity}\label{subsec:mmw_complexity}
We then discuss the computational routines of the MMW in Algorithm \ref{alg:mmw-im} and their complexity.
First, formatting the canonical form requires computing the norm of the constraint coefficient matrices, which is $\mathcal{O}(K\Omega)$ complexity by processing only the non-zero elements in $\mathbf{A}^{(c)}$ (since $\mathbf{A}^{(c)}$ have total $\mathcal{O}(K\Omega)$ non-zero elements $\forall c$, as defined in appendix.
Then, in each iteration of MMW, constraint weights $\mathbf{y}^{[n]}$ in \eqref{eq:mmw:vector_exponential} is constructed by computing $\mathbf{A}^{(c)}\bullet\mathbf{X}^{[n]}$ in the current iteration and added to the accumulated values in previous iterations, $\sum_{n'<n}\mathbf{A}^{(c)}\bullet\mathbf{X}^{[n']}$, $\forall c$.
When computing $\mathbf{A}^{(c)}\bullet\mathbf{X}^{[n]}$ $\forall c$, we first evaluate the non-zero off-diagonal elements in all constraints, $\sum_{i,j}A^{(c)}_{i,j}X^{[n]}_{i,j}$, $\forall A^{(c)}_{i,j}\neq0$ and $i\neq j$, which costs $\mathcal{O}(K\Omega)$ complexity. 
For diagonal elements in constraints $c=1\dots K$, $A^{(c)}_{k,k}$ are the same for all $k$ except $A^{(c)}_{c,c}$.
Further, due to $\sum_k X^{[n]}_{k,k} = K$, the diagonal elements $\sum_{k}A^{(c)}_{k,k}X^{[n]}_{k,k}$, $c=1\dots K$, can be computed in $\mathcal{O}(K)$ using the distributive law that combines sums. 
For diagonal elements in constraints $c=K+1,\dots,C$, the computation has $\mathcal{O}(C)\approx \mathcal{O}(K\Omega)$ complexity since they are scaled identity matrices as defined in the appendix.
Similarly, computing the loss matrix in \eqref{eq:mmw:loss_matrix} has $\mathcal{O}(K\Omega)$ using distributive law that combines the computation of the identical elements.
When approximating the matrix exponential in \eqref{eq:mmw:matrix_exponential}, we first add the loss matrix in the current iteration to the sum of those in the previous iteration as $\sum_{n'<n}\mathbf{L}^{[n']} + \mathbf{L}^{[n]}$, which is $\mathcal{O}(K\Omega)$ complexity by adding only non-zero elements.
\eqref{eq:mmw:expm_approximation} is computed based on Taylor's expansion on the matrix exponential \cite{al2010new} that is approximately $D$ times the complexity at the number of non-zero elements in $\sum_{n'=1}^{n}\mathbf{L}^{[n']}$, i.e., $\mathcal{O}(D \cdot K\Omega)$.
The inner products $\langle \mathbf{v}^{[n]}_i,\mathbf{v}^{[n]}_j\rangle$ in \eqref{eq:mmw:inner_product_of_matrix_exponential}, \eqref{eq:mmw:trace} and \eqref{eq:mmw:matrix_exponential_vector_representaion} are only computed for $i,j$ where $\exists c$ $A^{(c)}_{i,j}\neq0$, i.e., with a complexity at $\mathcal{O}(D \cdot K\Omega )$.
The average of the PSD matrices and constraint weights are computed by first adding their current iteration's values to ones in previous iterations and then averaging them when iterations finish. It requires $\mathcal{O}(K\Omega)$ for both adding the PSD matrix, i.e., adding only elements computed in the inner products, and adding constraint weights. Finally, the TSVD in \eqref{eq:mmw:tsvd_x_bar} and \eqref{eq:mmw:gram_x_bar} has $\mathcal{O}(D \cdot K\Omega)$ complexity using iterative eigenvalue algorithms \cite{lehoucq1998arpack}, where $\mathcal{O}(K\Omega)$ is the number of computed elements in $\Bar{\mathbf{X}}$ in above steps, while other elements in $\Bar{\mathbf{X}}$ are set to $0$.

In summary, each iteration of Algorithm \ref{alg:mmw-im} can be implemented at $\mathcal{O}(D \cdot K\Omega) \approx \mathcal{O}(K\Omega^2)$ and the whole algorithm is at complexity $\mathcal{O}(N K\Omega^2)$, where $K$, $\Omega$ and $N$ are the number of users, the maximum number of neighbors in interference graphs and the number of MMW iterations, respectively. By substituting the MMW complexity into the complexity derived in Section \ref{subsec:sdp_framework_complexity}, the SIG-SDP framework has complexity $\mathcal{O}(N K\Omega^2\log \Omega)$ when using the MMW as the SDP solver. Later, in simulations, we will show that $N$ can be set to a constant while ensuring convergence performance regardless of network sizes. Moreover, simulations will show that $\Omega$ converges to a constant in large networks. In other words, the framework's complexity is approximated at a near-linear complexity w.r.t. the number of users, i.e.,  $\mathcal{O}(K)$.

\section{Online Architecture of SIG-SDP}\label{sec:online_architecture}
\begin{figure}[!t]
\centering
\includegraphics[scale=0.75]{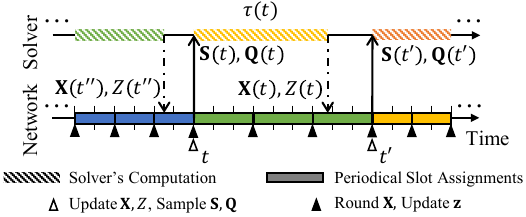}
\caption{Interaction between the solver and the network over time.}
\label{fig:online_architecture}
\vspace{-0.4cm}
\end{figure}

In this section, we explain the online architecture that uses the proposed SIG-SDP framework to continuously adjust the slot number $Z$ and the slot assignment $\mathbf{z}$ for dynamic networks, e.g., due to users' mobility. The framework computes slot assignments according to online measured path gains and runs in parallel with the network, as shown in Fig. \ref{fig:online_architecture}. Without explicit notation, the PSD matrix $\mathbf{X}$ returned from the SDP solver is stored in its gram form $\gram(\mathbf{X})$ in this section.
Specifically, the network states, including the interference powers and user associations, and the slot assignments, including the integer or relaxed slot assignments and the slot number, are denoted as $\mathbf{S}(t)$, $\mathbf{Q}(t)$, $\mathbf{z}(t)$, $\mathbf{X}(t)$, and $Z(t)$ at slot $t$, respectively. The details on how the network states and slot assignments are measured and updated are as follows.

\subsection{Online Architecture Design}
At the initialization of the network, i.e., at $t=1$, the network measures the states, including the interference powers $\mathbf{S}(t)$ and the user associations $\mathbf{Q}(t)$ and sends them to the solver. The solver runs the framework in Algorithm \ref{alg:im-sdp_framework} based on these states to compute the relaxed slot assignments $\mathbf{X}(t)$ and the slot number $Z(t)$ based on given $\mathbf{S}(t)$ and $\mathbf{Q}(t)$.
The network waits until the solver finishes at slot $t'$.
We assume that the network makes no transmissions during the first computation of the solver in the initialization, i.e., from $t=1$ to $t'$.
Once slot $t'$ is reached, the initialization completes, and we update slot indices as
\begin{equation}\label{eq:online_slot_update}
\begin{aligned}
t'' \leftarrow t,\ t \leftarrow t'  ,
\end{aligned}
\end{equation}
where $t''$ denotes the slot when the previous solver computation started. Next, at current slot $t$, the network re-measures $\mathbf{S}(t)$ and $\mathbf{Q}(t)$ and the solver recomputes $\mathbf{X}(t)$ and $Z(t)$ as
\begin{equation}\label{eq:online_sdp}
\begin{aligned}
\mathbf{X}(t) ,  Z(t) \leftarrow \mathrm{SdpBinarySearch}(\mathbf{S}(t), \mathbf{Q}(t)),
\end{aligned}
\end{equation}
where we only take the PSD matrix and the slot number from the solver's returned values in Algorithm \ref{alg:im-sdp_framework}.
When the solver is computing $\mathbf{X}(t)$ and $Z(t)$, the network repeatedly rounds the slot assignments based on the previous returned $\mathbf{X}(t'')$ and $Z(t'')$ and the latest network states $\mathbf{S}(i)$ and $\mathbf{Q}(i)$ as
\begin{equation}\label{eq:online_rounding}
\begin{aligned}
\mathcal{V}^{1},\dots,\mathcal{V}^{Z} &\leftarrow \text{Round}\ \gram(\mathbf{X}(t'')) \ \text{as \eqref{eq:rounding_order},\eqref{eq:rounding_check_per_user}}, \\
z_k(i) &\leftarrow  z , \ \forall k \in \mathcal{V}^{z}, \ \forall z , \\
z_k(i) &\leftarrow  z \sim \mathcal{U}\{1,Z(t'')\}, \ \forall k \notin \cup_z \mathcal{V}^{z} ,
\end{aligned}
\end{equation}
where slot $i$ is the starting slot of each transmission period during the solver's current computation, i.e., $i=t,t+Z(t''), t+2Z(t''),\dots$. Here, we assign a random slot in $\mathcal{U}\{1,Z(t'')\}$ to users without a slot assigned (i.e., they have unsatisfied interference constraints due to the variation of the network). The network schedules users' transmissions according to the slot assignments in \eqref{eq:online_rounding}, i.e., user $k$ transmits in the slot $i+z_k(i)$ in the period $[i,i+Z(t''))$. The slot assignment process in \eqref{eq:online_rounding} repeats until the solver finishes computing $\mathbf{X}(t)$ and $Z(t)$, i.e., until slot $t'$ where
\begin{equation}\label{eq:online_next_time_slot}
\begin{aligned}
t' \leftarrow t+\lceil \tau(t)/(\Delta_\mathrm{0} Z(t'')) \rceil \cdot Z(t'').
\end{aligned}
\end{equation}
Here, $\tau(t)$ is the duration of the solver's computation started at $t$. Then, the processes \eqref{eq:online_slot_update}-\eqref{eq:online_next_time_slot} are repeated to continuously adjust slot assignments until the network terminates.

Overall, the proposed architecture updates the relaxed slot assignment $\mathbf{X}(t)$ and the slot number $Z(t)$ for each time the solver returns while updating the integer slot assignments $\mathbf{z}(i)\triangleq [z_1(i),\dots,z_K(i)]$ every transmission period in parallel. The shorter the solver's computing time, the more timely the relaxed slot assignment and the slot number are.

\begin{figure}[!t]
\centering
\includegraphics[scale=0.7]{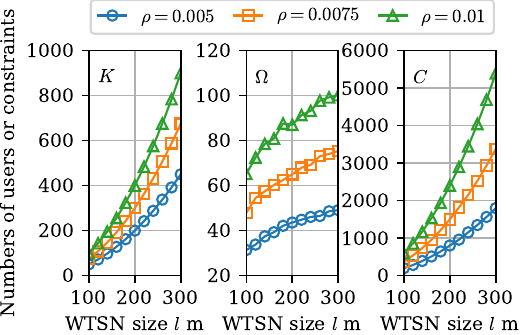}
\vspace{-0.1cm}
\caption{The number of users and constraints in simulated networks.}
\label{fig:plot_graph_test}
\vspace{-0.4cm}
\end{figure}

\subsection{Real-World Implementation Considerations}\label{subsec:real_world_implementation}
Although this paper focuses on the theoretical framework and algorithm design, we also discuss practical aspects of implementing the SIG-SDP framework in real-world systems.

First, we discuss how the control/measurement signalling should be implemented. The path gains are measured at the BSs when a user sends a data packet. This is feasible in practice as the path gains can be estimated from the pilot signals sent by users in the uplink control channels \cite{3gpp38211,3gpp38213}. BSs send the measured path gains to the controller. The controller can use these path gains to construct the interference graphs, $\mathbf{S}(t)$ and $\mathbf{Q}(t)$, and then run the designed solver to compute the slot assignments.
Once the solver returns the relaxed slot assignments $\mathbf{X}(t)$ and the slot number $Z(t)$, the controller can round the relaxed slot assignments to integer slot assignments, as described in \eqref{eq:online_rounding}.
At the start of each transmission period, the controller informs each BS about the slot assignments for their associated users and which users they should decode in each slot, according the rounding process in \eqref{eq:online_rounding}.
The BSs send the scheduling control information in a separate downlink control channel \cite{3gpp38211,3gpp38213}. This control information includes the integer slot assignments for each user, the slot number/periodicity, the power control commands to configure the transmission power as \eqref{eq:user_tx_power}. Then, the users can transmit their data packets in the following period accordingly.
The above signalling process is repeated over time.

Second, we discuss the time synchronization and signalling delay. The transmissions from users are required to be synchronized to align the slots. This can be achieved by first synchronizing the BSs using the precision time protocol (PTP) \cite{godor2020look}, which achieves sub-microsecond accuracy for a local area network \cite{gu2021knowledge}. The users are synchronized to the BSs using the downlink control channel, which can also achieve sub-microsecond accuracy \cite{3gpp38133}. Thus, sub-microsecond time synchronization among users' transmissions can be established, sufficient for aligning the slot duration ($0.1\sim1$ millisecond) in typical WTSN systems \cite{3gpp.38.825}. Furthermore, the signalling delay between the BSs and the controller can be minimized by running the controller at the edge of the network, where the signalling delay is typically less than a few milliseconds \cite{gu2021knowledge}, i.e., negligible compared to the solver's computing time.

\begin{figure}[!t]
\centering
\includegraphics[scale=0.75]{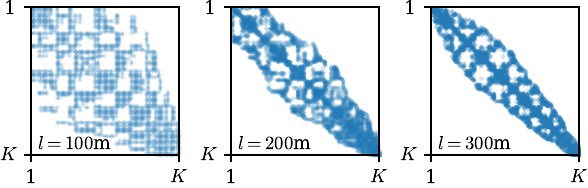}
\vspace{-0.1cm}
\caption{Sparsity patterns of non-zero elements in constraint coefficient matrices in a random realization of the network for different network sizes $l=100$, $200$ and $300$m when $\rho=0.0075$, i.e., $K=75$, $300$ and $675$, respectively. Here, users are permuted using reverse Cuthill McKee ordering \cite{cuthill1969reducing} to cluster non-zero elements around the diagonal.}
\label{fig:plot_matrix_sparsity}
\vspace{-0.4cm}
\end{figure}

\section{Simulation Results}\label{sec:simulation_results}
This section presents simulations evaluating our methods.
\subsection{Simulation Setups}\label{subsec:simulation_setup}
We assume that the WTSN system \cite{3gpp.38.825} is deployed in a $l$m~$\times$~$l$m cellular network centered at $(l/2,l/2)$ meters.
Each cell is a 20m~$\times$~20m square with one BS located at the center, i.e., there are $A = l^2/20^2$ cells (or BSs). 
The $(a_\mathrm{x},a_\mathrm{y})$-th cell is centered at $(20a_\mathrm{x}- 10,20a_\mathrm{y}-10)$, where $a_\mathrm{x}$ and $a_\mathrm{y}$ are both equal to $1,\dots,l/20$.
Unless specifically stated, we assume that users are static, and we set the number of users $K$ based on the user density $\rho$ (users per square meter) as $K=\rho l^2$.
Each user is randomly distributed in the simulated area, where each coordinate is randomly selected from the uniform distribution with an interval $[0,l]$ meters.
The slot duration of the devices is configured as $\Delta_\mathrm{0}=0.125$~millisecond. The bandwidth $B$ is set as $5$~MHz at a $4$~GHz carrier frequency. 
The packet size $L$ is 800 bits, and the maximum decoding error rate $\epsilon^{\max}$ for each user is $10^{-5}$ when no interference exists.
The additional transmission power offset is set as $\alpha=1$, and the noise power $\mathbb{N}_\mathrm{0}B $ is set as $-94$~dBm in decibel scale.
The path gains between a user and a BS follow a log-distance path loss model \cite{series2015propagation} as $-28\log_{10}(d+1)-20\log_{10}(f/(1\text{MHz}))+12$ in decibel scale, where $f$ is the carrier frequency in MHz and $d$ is the user-to-BS distance in meters.
In the simulations, all interference, including unmeasurable interference, is considered when evaluating the actual packet loss $\epsilon_k$ for each user $k$.
To reduce the impact of unmeasurable interference, we assume highly sensitive receivers at the BSs, i.e., the signal strength of a user at the BS is measurable if signal strength of the user at the BS is higher than $-10$dB of the noise power, i.e., $\gamma=0.1$. In practical systems, this extreme value of sensitivity is achievable by using advanced receiver designs \cite{ngo2011uplink}, while the typical sensitivity threshold may be higher \cite{3gpp38104}.
When multiple users associated with the same BS are assigned to the same slot, the BS decodes the user with the highest SINR, and the other users' packets will be discarded, i.e., their packet error rates are $1$.
We set $D=\lceil\beta (Z-1)\rceil$ where $\beta=2$ for each given $Z$ in the MMW, and the configuration of the update rate $\eta$ and the number of iterations $N$ will be discussed later. The MMW algorithm is implemented with SciPy sparse matrix library \cite{scipy} and runs in a virtual machine that is equipped with $6$ central processing unit cores at a 3GHz clock rate.

\subsection{Evaluation of Sparsity of Graphs}
Fig. \ref{fig:plot_graph_test} illustrates the number of users and constraints in simulated networks with different user densities. It shows that the number of users, $K$, and the number of constraints, $C$, grow at a close rate as the network size increases, validating the linearity between them in \eqref{eq:linearity_in_K_and_C}. The maximum number of neighbors for users, $\Omega$, grows at a decreased rate as the network size increases. This indicates the maximum number of neighbors for users converges due to the limited receiver sensitivity, i.e., the neighbors of a user are all located in a bounded area centered at each user's associated BS \cite{penrose2003random}.
The higher user density will lead to an increasing number of users and an increasing number of constraints as more neighboring users will associate with the same BS and interfere with the detection range of the BSs. In the rest of the evaluations, we fix user density as $\rho=7.5\times 10^{-3}$ users per square meter and vary the network size $l$ from $100$ to $300$ meters, i.e., the number of users $K$ varies from $75$ to $675$. Fig. \ref{fig:plot_matrix_sparsity} shows the sparsity patterns of the non-zero elements in the constraint coefficient matrices, i.e., the non-zero elements in interference graphs' adjacency matrices. It implies that most of the elements in these matrices are zero, and we can ignore the processing of those elements and the corresponding PSD matrix's elements to accelerate the solver.

\begin{figure*}[!t]
\centering
\subfloat[The convergence of the normalized duality gap, $\mathrm{gap}(\Bar{\mathbf{y}},\Bar{\mathbf{X}})$.]{\includegraphics[scale=0.75]{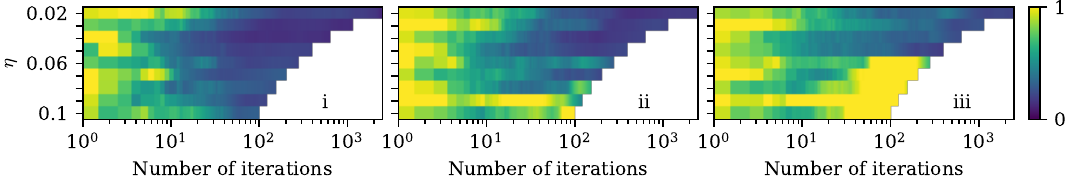}\label{fig:plot_duality_gap}}
\\
\subfloat[The convergence of the primal side of the gap, $\max_c \mathbf{A}^{(c)} \bullet \Bar{\mathbf{X}}$.]{\includegraphics[scale=0.75]{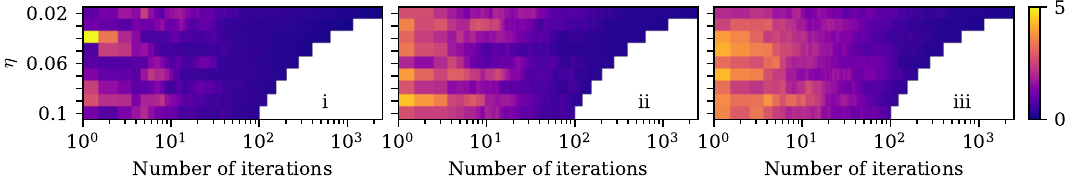}
\label{fig:plot_max_vio}}
\caption{The convergence of the duality gap in the MMW, $\mathrm{gap}(\Bar{\mathbf{y}},\Bar{\mathbf{X}})$, in \eqref{eq:definition:sdp_canonical_form_duality_gap} when i) $l = 100$m, ii) $l = 200$m and iii) $l = 300$m.}
\label{fig:plot_duality}
\vspace{-0.4cm}
\end{figure*}

\subsection{Baseline Methods}
We then explain the implementation of the baselines different from the proposed methods, including other SDP solvers, heuristics and linear relaxation\footnote{Machine learning and spectral graph methods are not directly applicable for the interference management problem formulated in this paper, as mentioned in Section \ref{subsec:related_works}, and we discuss how to apply them in the appendix.}.

\subsubsection{SDP Solvers (ADMM, RAND)}\label{subsec:sim_res:base_line:sdp_solvers}
We compare the PDIP and the ADMM solvers with the MMW solver in the framework. 
Note that we observe that the PDIP solver \cite{toh1999sdpt3,yamashita2003implementation} fails to scale to the problem with $\geq 100$ users for our simulated scenarios (this fact is also observed in \cite{shi2015largeA}) due to its high computational complexity, and it is not compared in simulations.
The ADMM solver \cite{o2016conic} is referred to as the ``ADMM'' scheme. Further, we compare the case where the SDP solver simply returns a random PSD matrix, referred to as the ``RAND'' scheme, to illustrate that the returned PSD matrices from the MMW and ADMM SDP solvers are meaningful.

\subsubsection{Heuristic Algorithms (MINTP, MASSO)}
We compare our methods with heuristics based on greedy max-weight scheduling policies \cite{tassiulas1990stability}. Specifically, users are first sorted based on a weight $w_k$ for each user $k$, $k=1,\dots,K$, as $k_1,\dots,k_K$, where $w_{k_1}\geq w_{k_2}\geq\dots \geq w_{k_K}$. We then try to assign users from $k_1$ to $k_K$ to the same slot. If adding the user to the slot does not cause a violation of interference constraints for the previous users and the user itself, the user is assigned to the slot and is marked as assigned; otherwise, the user remains unassigned, and we move to the next user. After the slot is tested with all users, we add a new slot and try to assign users to it in the same order as above, while those users assigned with a slot before are removed. The process is repeated until all users are assigned. We test the following weighting schemes: 1) weighting based on interference power to the user, i.e., $w_k=\sum_{k'\neq k} S_{k',k}$, referred to as the ``MINTP'' scheme and 2) based on the number of users associated with the same BS, $w_k=\sum_{k'\neq k} Q_{k',k}$, referred to as the ``MASSO'' scheme.

\subsubsection{Linear Relaxation Method (LP)}
Additionally, we compare our SDP relaxation approach with the linear relaxation one \cite{subramanian2008minimum}.
We relax the binary slot assignments to values $x_{k,z}\in [0,1]$ indicates how likely user $k$, $k=1,\dots,K$, is assigned to slot $z$, $z=1,\dots,Z$ for the given slot number $Z$. 
As a user is only assigned with one slot, i.e., $\sum_z x_{k,z} = 1$. Users associated with the same BS should be assigned to different slots, i.e., $x_{k,z}+x_{k',z}\leq 1$ $\forall z,\ \forall Q_{k,k'} = 1$.
The interference of a user should be less than the threshold in each slot, i.e., $\sum_{k'\neq k} S_{k',k}x_{k',z}\leq (1-x_{k,z}) \sum_{k'\neq k} S_{k',k}  + x_{k,z} \alpha $ $\forall k,z$. Here, if user $k$ is less likely assigned to slot $z$ (i.e., when $x_{k,z}$ is lower), the interference threshold of user $k$ is higher, e.g., to $\sum_{k'\neq k} S_{k',k}$ when $x_{k,z}=0$ and all interference can be tolerant in the slot, since it is not likely transmitting in this slot $z$; otherwise, the threshold is lower, e.g., to the maximum allowed interference $\alpha$ as \eqref{eq:const:max_interference_original} when $x_{k,z}=1$. We solve the above linear relaxation using ADMM \cite{o2016conic}, which is referred to as the ``LP'' scheme.

\begin{figure}[!t]
\centering
\subfloat[Convergence of the MMW with different $\rho$.]{\includegraphics[scale=0.775]{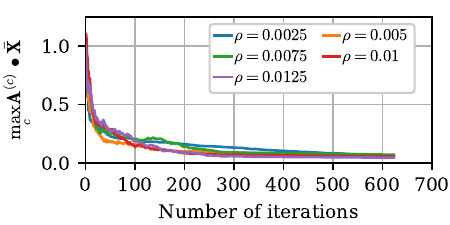}
\label{fig:plot_convergence_rho}}
\vspace{-0.3cm}
\\
\subfloat[Convergence of the MMW with different $\alpha$.]{\includegraphics[scale=0.775]{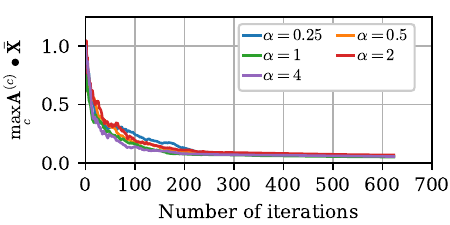}
\label{fig:plot_convergence_alp}}
\caption{
Convergence sensitivity of the MMW algorithm in maximum constraint violations when varying the sparsity level.}
\label{fig:plot_convergence}
\vspace{-0.4cm}
\end{figure}

\subsection{Convergence Analysis of MMW}\label{subsec:sim_res:convergence_analysis}

We evaluate the convergence of the MMW algorithm in Algorithm \ref{alg:mmw-im} for the given minimum number of slots. Specifically, we first run the framework with the MMW to find the minimum number of slots. 
Then, for the given minimum number of slots, we record the values of the dual gap in \eqref{eq:definition:sdp_canonical_form_duality_gap} for averaged primal and dual variables at different numbers of MMW iterations, i.e., $\mathrm{gap}(\Bar{\mathbf{y}},\Bar{\mathbf{X}})$.
We vary the step size of optimization variables, $\eta$, from $0.02$ to $0.1$ and set the number of iterations, $N$, as $1/\eta^2$ (ignoring the logarithm term in Theorem \ref{theorem:mmw_convergence} for simplicity).
Fig. \ref{fig:plot_duality_gap} shows the normalized values of the duality gap under different network sizes. Results indicate that when $\eta$ is low, e.g., $0.02\sim0.05$, the gap converges when the number of iterations is around $100\sim200$ for all tested network sizes. This implies the designed MMW can converge in a close number of iterations for different network sizes when $\eta$ is small, and we can set the same number of iterations for tested network sizes to solve the SDP. Results also show that the gap fails to converge for high $\eta$ in large networks with more users.

To better illustrate the convergence of the gap, Fig. \ref{fig:plot_max_vio} shows the primal side of the gap, $\max_c \mathbf{A}^{(c)} \bullet \Bar{\mathbf{X}}$, that indicates the maximum constraint violations. 
Results indicate that the maximum violations converge approximately at $100$ iterations. By comparing Fig. \ref{fig:plot_duality_gap} and Fig. \ref{fig:plot_max_vio}, we can identify that the dual side of the gap, $\lambda_{\min}(\sum_c \Bar{y}_c \cdot \mathbf{A}^{(c)})K$, fails to converge for high $\eta$ in large networks (for a concise presentation, the values of the dual side are not plotted). This is because $\mathrm{Softmax}(\cdot)$'s outputs (i.e., the dual variables) in \eqref{eq:mmw:vector_exponential} can hardly converge when it has a high dimension in large networks (due to a large number of constraints) and simultaneously has high update rates in its input (due to high $\eta$).
Nevertheless, without explicit statement, we set $\eta=0.04$ and $N=150$ in the remaining simulations since these values lead to the convergence of the gap in all tested network sizes, as shown in Fig. \ref{fig:plot_duality}.
We further evaluate the convergence of the MMW when the user density $\rho$ and transmission power offset $\alpha$ vary when $l=100$m, as shown in Fig. \ref{fig:plot_convergence_rho} and Fig. \ref{fig:plot_convergence_alp}, respectively. These parameters affect the number of interfering neighbors of each user, and therefore, it impacts the sparsity level of the interference graphs.
Results show that the MMW converges for different $\rho$ and $\alpha$ values at a similar number of iterations, i.e., $N=150$., which indicates the robustness of the MMW to different sparsity levels.

\subsection{Performance of Proposed SIG-SDP Framework}

\begin{figure}[!t]
\centering
\subfloat[The SIG-SDP framework with different $l$.]{\includegraphics[scale=0.75]{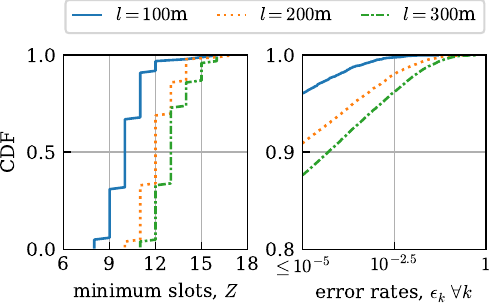}\label{fig:plot_data_slot_bler_mmw_different_cell_size}}
\vspace{-0.3cm}
\\
\subfloat[The framework using different SDP solvers.]{\includegraphics[scale=0.75]{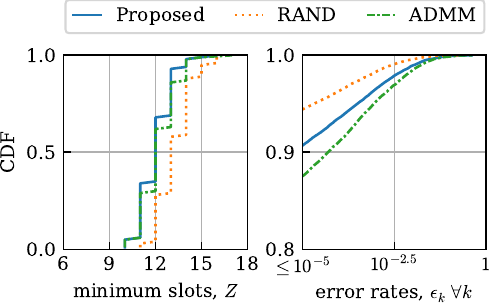}
\label{fig:plot_data_slot_bler_scs_mmw}}
\caption{CDFs of the minimum number of slots and the error rates when using the SIG-SDP framework.}
\label{fig:plot_framework_performance}
\vspace{-0.4cm}
\end{figure}

Fig. \ref{fig:plot_data_slot_bler_mmw_different_cell_size} shows the cumulative distribution functions (CDF) of the minimum number of slots and the packet error rates when using the proposed framework and the MMW algorithm as the SDP solver. We measure the CDFs in three network sizes, e.g., when $l = 100$m, $200$m and $300$m. Results show that the minimum number of slots increases as the network size increases due to growing numbers of neighboring users (as shown in Section \ref{subsec:simulation_setup}). Results also show that a proportion of users, approximately $5\%\sim10\%$, has a packet error rate violating the required threshold $10^{-5}$. This is because the algorithms cannot sense unmeasurable interference due to the limited receiver sensitivity, and the accumulation of unmeasurable interference causes lower SINRs than the required threshold. Furthermore, as the network size increases, the proportion of users with violated packet error rates becomes larger due to more users generating unmeasurable interference in larger networks.
Fig. \ref{fig:plot_data_slot_bler_scs_mmw} compares the CDFs when using the framework with the proposed MMW solver (with legend ``Proposed'') to the cases with ADMM and RAND solvers when $l=200$m.
Comparing the ADMM and the proposed MMW solvers, the proposed one uses slightly fewer slots ($\sim2\%$) and fewer users with violated packet error rates ($\sim3\%$). 
This is because the proposed MMW solver iterates the primal variables (the PSD matrix) in low ranks as \eqref{eq:mmw:matrix_exponential_vector_representaion}, which is a closer approximation to the integer slot assignment than the ones in ADMM without the rank restriction. 
Further, the RAND solver returning the random PSD matrix leads to $10\%$ more slots than the proposed one. This implies that the PSD matrix returned from the MMW contains the user correlation that can guide the efficient slot assignments when rounding. 
Meanwhile, the RAND solver has $\sim4\%$ fewer users with violated error rates than the proposed one, as unmeasurable interference becomes less in each slot when using more slots.

\begin{figure}[!t]
\centering
\subfloat[Comparison to RAND and LP schemes.]{\includegraphics[scale=0.75]{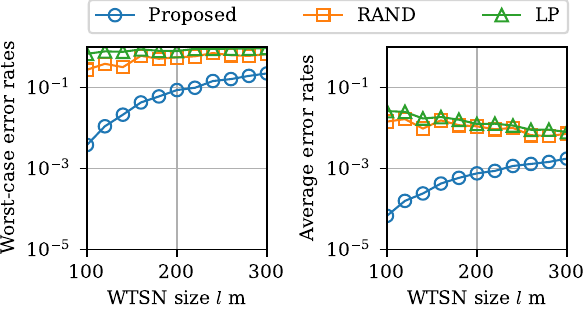}\label{fig:plot_data_bler_avg_max_mmw_lp_rand}}
\vspace{-0.3cm}
\\
\subfloat[Comparison to heuristics.]{\includegraphics[scale=0.75]{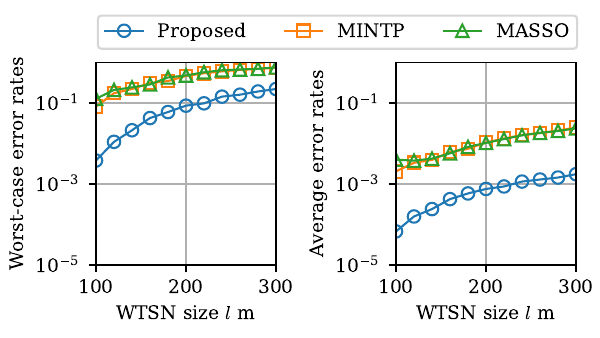}
\label{fig:plot_data_bler_avg_max_mmw_heuristic}}
\caption{User packet error rates when using different schemes.}
\label{fig:plot_data_bler_avg_max}
\vspace{-0.4cm}
\end{figure}

Fig. \ref{fig:plot_data_bler_avg_max} compares the user packet error rates of different schemes when the slot numbers in each scheme are the same, where we vary the network size from $l=100$m to $300$m. We first use the SIG-SDP framework to find the minimum number of slots and measure the error rates for the given slot number in different schemes.
Fig. \ref{fig:plot_data_bler_avg_max_mmw_lp_rand} compares the proposed MMW solver with the random SDP solver and the linear relaxation schemes, i.e., the RAND and LP schemes. Results show that the proposed solver achieves approximately $100$ times lower packet error rates when the network size is small. 
This is because the correlation between each pair of users' slot assignments is not optimized in the RAND and LP schemes, while our scheme optimizes the PSD matrix guiding the slot assignments when rounding. Results also show the difference between the proposed solver and the compared ones is smaller when the network sizes increase. This is due to the large unmeasurable interference that increases the error rates. 
Fig. \ref{fig:plot_data_bler_avg_max_mmw_heuristic} shows that the proposed MMW solver achieves $10$ times lower error rate than heuristics since heuristics do not exploit the correlation between users' slot assignments.

\subsection{Computational Overheads in SIG-SDP Framework}\label{subsec:sim_res:computing_time}
Next, we evaluate the computational overheads of the SIG-SDP framework. We first measure the real-world computing time of the MMW solver and then compare the computing time of the framework with the MMW solver to the other methods. Finally, we evaluate the impact of the computing time on the performance of the framework in dynamic networks when we deploy the framework in the online architecture.

\begin{figure}[!t]
\centering
\subfloat[Computing time of the MMW solver.]{\includegraphics[scale=0.75]{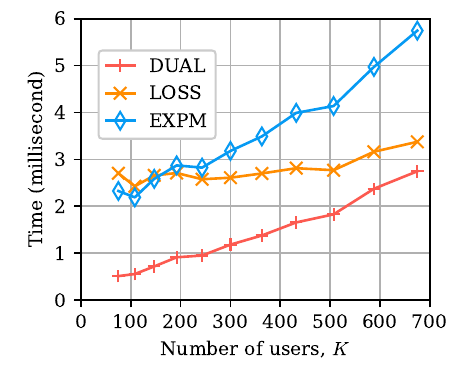}\label{fig:plot_data_mmw_time}}
\vspace{-0.3cm}
\\
\subfloat[Computing time of the framework.]{\includegraphics[scale=0.75]{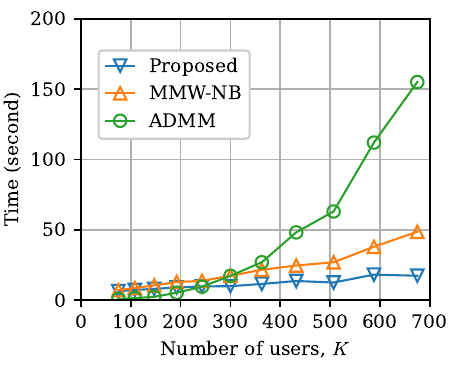}
\label{fig:plot_data_mmw_scs_iter_time}}
\caption{Computing time of the MMW solver, the SIG-SDP framework and other comparison schemes when $l$ varies from $100$m to $300$m, i.e., the number of users $K$ varies from $75$ to $675$.}
\label{fig:computing_time}
\vspace{-0.4cm}
\end{figure}

\subsubsection{Computational Overhead of MMW}
We first evaluate the computing time of the MMW solver given different network size from $l=100$m to $300$m, and consequently, the number of users $K$ varies from $75$ to $675$.
 Specifically, Fig. \ref{fig:plot_data_mmw_time} shows the computing time of each MMW iteration when updating the dual variables in \eqref{eq:mmw:vector_exponential} (with legend ``DUAL''), updating the loss matrix in \eqref{eq:mmw:loss_matrix} (with legend ``LOSS'') and approximating the matrix exponential \eqref{eq:mmw:matrix_exponential} in  \eqref{eq:mmw:expm_approximation}-\eqref{eq:mmw:matrix_exponential_vector_representaion} (with legend ``EXPM''). 
Here, the collected computing time is measured when the minimum number of slots is given, i.e., in the last iteration of the binary search in the SIG-SDP framework.
Results show that the above components cost computing time linearly scaling with the number of users, validating the complexity analysis of the MMW in Section \ref{subsec:mmw_complexity}.

\subsubsection{Computational Overhead Comparisons}
Moreover, with increasing network size and the number of users, Fig. \ref{fig:plot_data_mmw_scs_iter_time} shows the SIG-SDP framework's computing time when using the MMW algorithm as the solver, both within the derived slot bounds from Theorem \ref{theorem:min_slot_bounds} (with legend ``Proposed'') and without bounds by simply searching in $[1,K]$ (with legend ``MMW-NB''). Also, it compares the ADMM-based framework  (with legend ``ADMM'') using the derived bounds.
Results show that the MMW solver reduces the computing time of the framework by up to $10$ times in large networks compared to the classic ADMM solver. The MMW-based framework shows near-linear growth in computing time w.r.t. the number of users due to its near-linear complexity, leveraging the sparsity of interference graphs by ignoring non-interfering user pairs. In contrast, ADMM has polynomial complexity, resulting in significant computing time in large-scale networks. Additionally, using bounds on interference graphs' chromatic numbers reduces computing time by up to $50\%$ by limiting the search range closer to the minimum.

\subsubsection{Impact of Computational Overhead in Dynamic Networks}\label{subsubsec:computing_time_dynamic_networks}
Finally, we evaluate the impact of the computing time of the SIG-SDP framework on the performance in dynamic networks.
We use the online architecture to apply the SIG-SDP framework (with legend ``Proposed'') in dynamic networks with mobile users and compare the performance with the MINTP heuristic scheme (with legend ``Heuristic'') that has a very low computing time ($\sim40$ milliseconds) and the framework using the ADMM solver (with legend ``ADMM'') with a longer computing time.
We also simulate an ideal case where the framework costs no computing time (with legend ``Ideal'').
The network size is set as $l=200$m, and users are moving in a random direction, where the speed varies from $0.2$ to $1$ meters per second, e.g., at approximately the walking speeds of humans or robots. Note that when a user reaches the boundary of the simulated network, it will move in another random direction towards the inside of the network. Results in Fig. \ref{fig:plot_data_bler_online_methods} show that the proposed scheme achieves lower packet error rates than the heuristic scheme when users move slowly. 
However, its performance advantage diminishes as user speed increases. 
This is because the proposed scheme requires more time to determine the slot assignments than the heuristic scheme. Consequently, in high mobility networks, user interference states have changed substantially by the time the proposed scheme returns the assignments.
The difference between the realistic and the ideal cases shows that the framework's performance can be improved by further accelerating its computation. It also shows that the ADMM-based framework has a high error rate due to its long computing time, which is not suitable for dynamic networks.
Additionally, we compare the MMW-based framework with different number of iterations in Fig. \ref{fig:plot_data_bler_online_iterations}. Results show that when users move slowly, the framework with the small or large number of iterations achieves higher error rates, as these cases either return the slot assignments too early without sufficient optimization or return the slot assignments too late to adapt to the dynamic interference changes, respectively. When users move at a higher speed, the framework with a few iterations achieves lower error rates than the one with many iterations, as the latter cannot adapt to the rapid changes in interference states due to the long computing time.


\begin{figure}[!t]
\centering
\subfloat[Comparison with different methods.]{\includegraphics[scale=0.75]{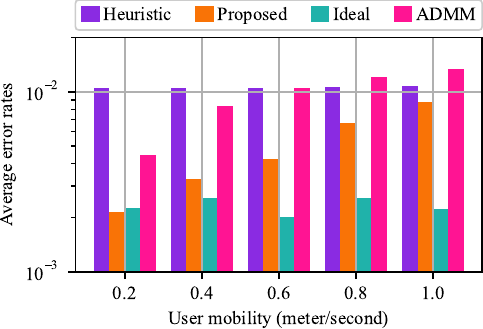}
\label{fig:plot_data_bler_online_methods}}
\vspace{-0.2cm}
\\
\subfloat[Comparison with different iterations.]{\includegraphics[scale=0.75]{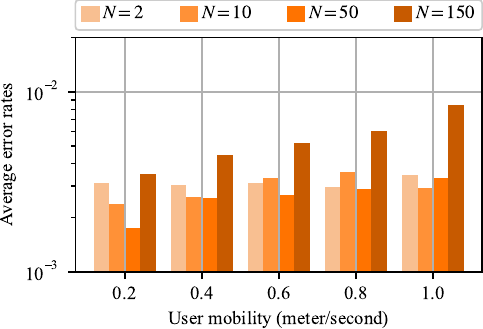}
\label{fig:plot_data_bler_online_iterations}}
\caption{
Average packet error rates in dynamic networks with mobile users.}
\label{fig:plot_data_bler_online}
\vspace{-0.4cm}
\end{figure}

\section{Conclusion}\label{sec:conclusion}
This paper proposed the SIG-SDP framework for slot assignments to manage interference in large-scale WTSN to reduce the number of slots used while ensuring high reliability.
We have shown how to exploit the sparsity of interference graphs to reduce the complexity of solving the problem, e.g., using the chromatic numbers of sparse graphs to bound the minimum number of slots, and using sparse random sketching to approximate the optimal PSD matrix.
Consequently, the framework has a near-linear complexity as the number of users increases, which is more scalable than the existing polynomial-complexity SDP methods.
However, we have shown that the framework still requires a larger computing time than simple heuristics when returning the slot assignments. Consequently, the out-performance of the proposed methods compared to the heuristics, e.g., the reliability improvement, vanishes when the interference in the network varies fast.

Future research could be conducted on deep learning and spectral graph methods that perform one-shot mapping from the graphs directly to the minimum number of slots and the optimal PSD matrix, reducing the computing time. For example, observe that the iterated PSD matrix in the MMW is a matrix exponential of a linear combination of the constraint coefficient matrices.
Thus, neural networks can be trained to infer a near-optimal linear combination (dual variables).
Furthermore, we can treat the combined constraint coefficient matrices as a graph adjacency matrix and use spectral graph methods to approximate the gram matrix of the optimal PSD matrix.
As the trained neural networks and spectral graph methods returns with a single round of computation, they can further reduce the computing time. However, the exact design of this integration is still an open question.
Additionally, this work considered a channel model without fading, and path gains can be deterministically measured. When considering fading channels, the edges of graphs could be stochastic and hardly be measured explicitly. In this case, how to construct stochastic graphs and exploit their sparsity will be a challenge.
Moreover, the current online architecture waits the whole framework's process to be completed before returning the solution. Further study is needed on the design on the numerically stable and convergent algorithm with timely tracking on dynamic interference while concurrently computing the solution.
Finally, the current design of the SIG-SDP framework does not consider adaptive transmission power control. Due to the impact of unmeasurable interference as observed in the simulations, decreasing powers of users that are not heavily interfered can further reduce interference and improve reliability, particularly when the network size is large with many users. Integrating the MMW-based SDP framework for slot assignments with power allocation schemes to manage unmeasurable interference requires further research.

\bibliography{main}
\bibliographystyle{IEEEtran}

\begin{IEEEbiography}
[{\includegraphics[width=1in,height=1.25in,clip,keepaspectratio]{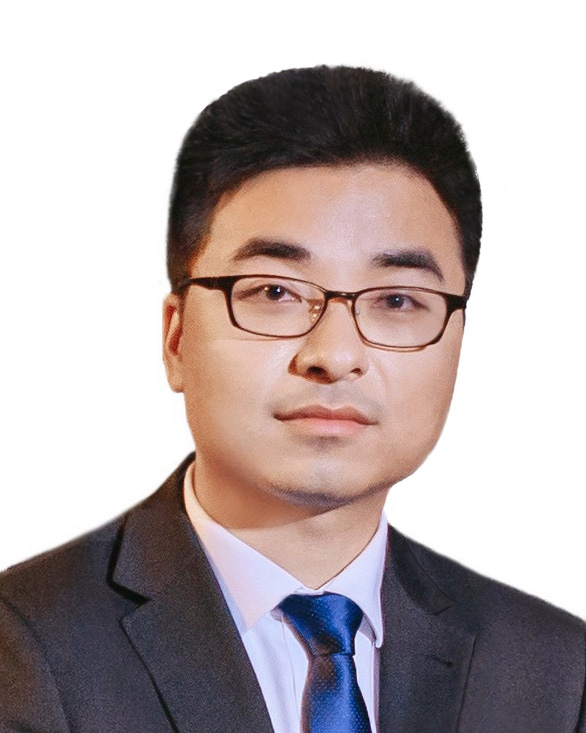}}]{Zhouyou Gu} completed his Ph.D. degree in 2023 at the School of Electrical and Information Engineering, the University of Sydney (USYD), Australia.
He was a research fellow at Deakin University, Australia in 2024. He is currently a research fellow at Singapore University of Technology and Design (SUTD), Singapore.
He was a recipient of the Australian Research Training Program Stipend, as well as the USYD Postgraduate Research Supplementary and Completion Scholarships, the NVIDIA Academic Grant Program, and the Best Student Paper Award at ML4Wireless@AAAI 2026.
His research focuses on making wireless networks more intelligent at scale through graph learning, reinforcement learning, and AI-native protocol design.
\end{IEEEbiography}

\begin{IEEEbiography}
[{\includegraphics[width=1in,height=1.25in,clip,keepaspectratio]{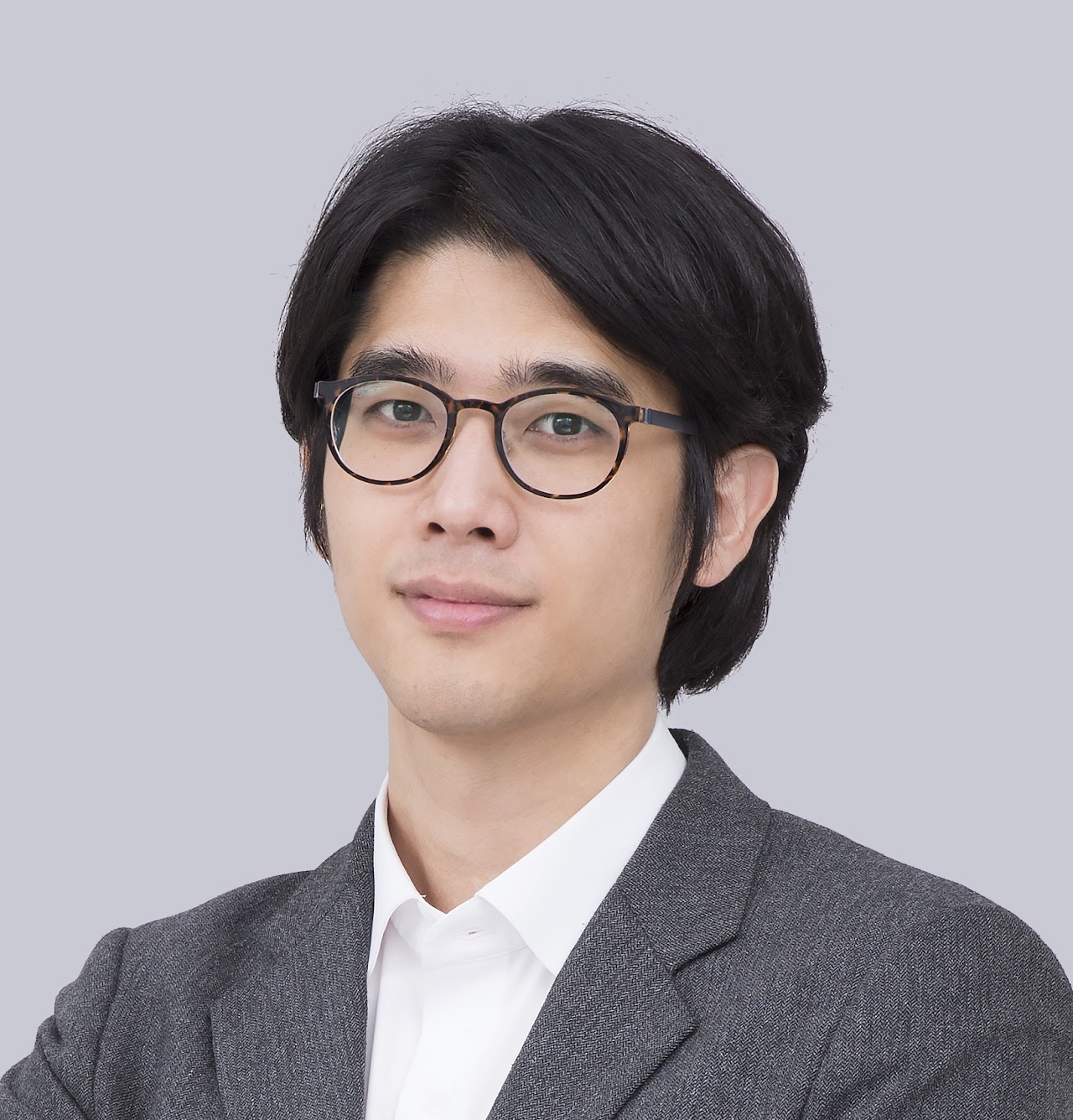}}]{Jihong Park} is an Associate Professor at the Singapore University of Technology and Design (SUTD) and an Honorary Associate Professor at Deakin University. He is Director of the MediaTek-SUTD Joint Laboratory and Deputy Director of the Future Communications Research and Development Programme (FCP) in Singapore. Dr. Park received his B.S. and Ph.D. degrees from Yonsei University, Seoul, South Korea, in 2009 and 2016, respectively. He was a Post-Doctoral Researcher with Aalborg University, Denmark, from 2016 to 2017, and the University of Oulu, Finland, from 2018 to 2019. His current research focus includes AI-native semantic communication and distributed machine learning for 6G. Dr. Park has served as Track Chair and Workshop Chair at leading conferences, including IEEE WCNC 2026, IEEE GLOBECOM 2023, and the 2025 ICML Workshop on Machine Learning for Wireless Communication and Networks. He has received several prestigious awards, including the 2023 IEEE Communication Society Heinrich Hertz Award and the 2022 FL-IJCAI Best Paper Award. Currently, Dr. Park is an Editor of IEEE Transactions on Communications and a Member of IEEE Signal Processing Society Machine Learning for Signal Processing Technical Committee, and Vice Chair of the AI-RAN Alliance’s AI-on-RAN Working Group.
\end{IEEEbiography}

\begin{IEEEbiography}
[{\includegraphics[width=1in,height=1.25in,clip,keepaspectratio]{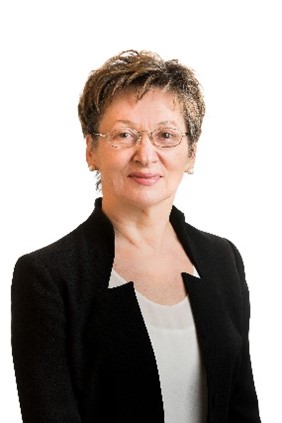}}]{Branka Vucetic} is an Australian Research Council Laureate Fellow and Director of the Centre of Excellence for IoT and Telecommunications at the University of Sydney. Her current research work is in wireless networks and the Internet of Things. In the area of wireless networks, she works on communication system design for millimetre wave frequency bands. In the area of the Internet of Things, Vucetic works on providing wireless connectivity for mission critical applications. Branka Vucetic is a Fellow of IEEE, the Australian Academy of Technological Sciences and Engineering and the Australian Academy of Science.
\end{IEEEbiography}

\begin{IEEEbiography}
[{\includegraphics[width=1in,height=1.25in,clip,keepaspectratio]{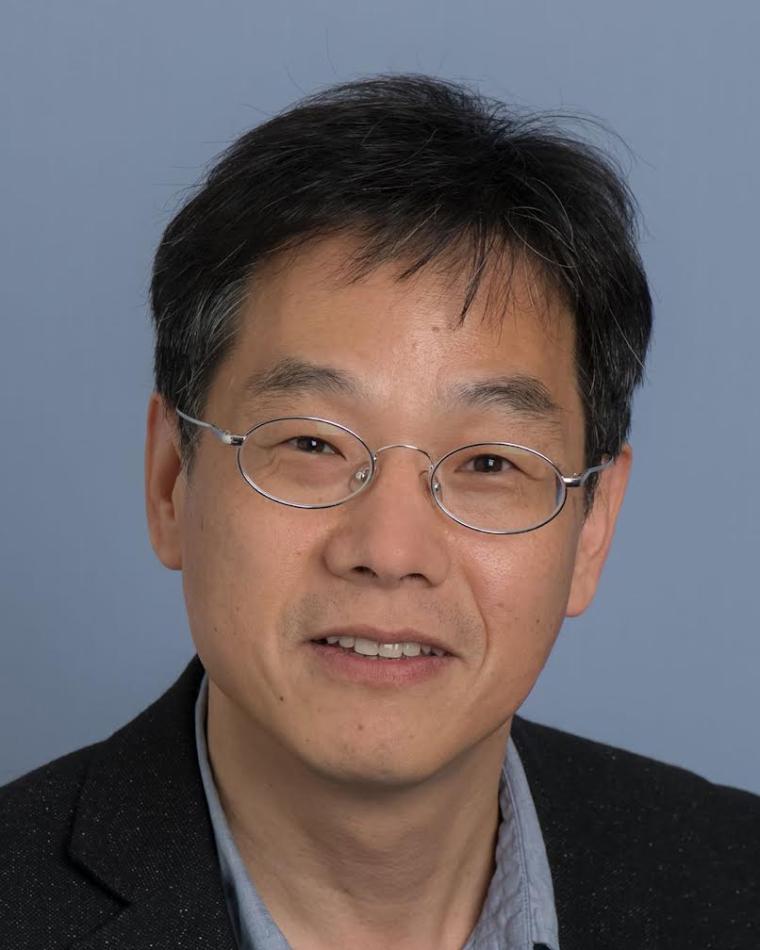}}]{Jinho Choi} was born in Seoul, Korea. He received B.E. (magna cum laude) degree in electronics engineering in 1989 from Sogang University, Seoul, and M.S.E. and Ph.D. degrees in electrical engineering from Korea Advanced Institute of Science and Technology (KAIST) in 1991 and 1994, respectively. He is with the School of Electrical and Mechanical Engineering, the University of Adelaide, Australia, as a Professor. His research interests include the Internet of Things (IoT), wireless communications, and statistical signal processing. He authored two books published by Cambridge University Press in 2006 and 2010 and one book by Wiley-IEEE in 2022. Prof. Choi received a number of best paper awards including the 1999 Best Paper Award for Signal Processing from EURASIP. He is a Fellow of the IEEE and has been on the list of World’s Top 2\% Scientists by Stanford University since 2020. Currently, he is a Senior Editor of IEEE Wireless Communications Letters and an Associate Editor of IEEE Trans. Mobile Computing. He has also served as a Division Editor of Journal of Communications and Networks (JCN), an Associate Editor or Editor of other journals including IEEE Trans. Communications, IEEE Communications Letters, JCN, IEEE Trans. Vehicular Technology, and ETRI journal.
\end{IEEEbiography}

\end{document}


\maketitle

\appendices

\section*{Appendix: The Proof of Theorems \ref{theorem:min_slot_bounds}}
We first prove the upper bound. We consider the greedy coloring algorithm \cite{west2001introduction} that colors the interference-power graph $\mathcal{G}^{\mathrm{intp}}$. The algorithm lists the users in an arbitrary order and assigns each user to a color that exists in the graph but is not occupied by its neighbors (either in or out neighbor). 
When the existing colors in the graph are all occupied by a user's neighbors, the algorithm adds a new color. 
Greedy coloring leads to a color scheme with at most $\Omega+1$ colors on $\mathcal{G}^{\mathrm{intp}}$ \cite{west2001introduction}. 
Furthermore, this color scheme leads to the fact that all neighbors in the association graph $\mathcal{G}^{\mathrm{asso}}$ have different colors. 
This is because each neighbor of a user in $\mathcal{G}^{\mathrm{asso}}$ is also a neighbor in $\mathcal{G}^{\mathrm{intp}}$. By setting the number of used colors and the color assignments in the greedy coloring algorithm as the slot number and the slot assignments, we obtain an interference-free slot assignment scheme that is a feasible solution to the problem in \eqref{eq:prob:im_problem_original}. As the minimum number of slots $Z^*$ are less than or equal to any feasible $Z$ and $Z\leq \Omega+1$, $\Omega+1$ is an upper bound of $Z^*$. We then prove the lower bound in Theorem \ref{theorem:min_slot_bounds}. To achieve this, we consider the graph coloring of $\mathcal{G}^{\mathrm{asso}}$. Because of the constraint in \eqref{eq:const:same_bs_diff_slot_original} (i.e., users associated with the same BS should be assigned to different slots.), feasible slot assignments in \eqref{eq:prob:im_problem_original} can be mapped to a graph coloring scheme of $\mathcal{G}^{\mathrm{asso}}$ that all neighbors have different colors. Using the Hoffman bound \cite{hoffman2003eigenvalues}, the minimum number of colors (or slots) needed is lower bounded by $1-\lambda_{\max}(\mathbf{Q})/\lambda_{\min}(\mathbf{Q})$, proving the lower bound.

\section*{Appendix: The Proof of Theorem \ref{theorem:duality_of_the_canonical_form}}
Since $\mathcal{X}$ and $\mathcal{Y}$ are compact and convex and $\sum_c y_c \cdot \mathbf{A}^{(c)} \bullet \mathbf{X}$ is linear in $\mathbf{y}$ and $\mathbf{X}$, the minmax theorem states $\mathcal{J}=\max_{\mathbf{y}\in \mathcal{Y}} \min_{\mathbf{X}\in \mathcal{X}}\sum_c y_c \cdot \mathbf{A}^{(c)} \bullet \mathbf{X}$.
Thus, 1) $\mathcal{J} = \min_{\mathbf{X}\in \mathcal{X}} \max_c \mathbf{A}^{(c)} \bullet \mathbf{X}$ due to the fact that $|\mathbf{y}|=1$; 2) $\mathcal{J}=\max_{\mathbf{y}\in \mathcal{Y}} \lambda_{\min}(\sum_c y_c \cdot \mathbf{A}^{(c)})K$ due to the minimum eigenvalue problem.
From 1) and 2), we have $\lambda_{\min}(\sum_c y_c \cdot \mathbf{A}^{(c)})K\leq\mathcal{J}\leq\max_c \mathbf{A}^{(c)} \bullet \mathbf{X}$, for all $\mathbf{X}\in \mathcal{X}$ and $\mathbf{y}\in \mathcal{Y}$, i.e., the duality gap in \eqref{eq:definition:sdp_canonical_form_duality_gap} is always non-negative.
Additionally, let $\mathbf{X}^* = \arg\min_{\mathbf{X}\in \mathcal{X}} \max_c \mathbf{A}^{(c)} \bullet \mathbf{X}$ and $\mathbf{y}^*= \arg \max_{\mathbf{y}\in \mathcal{Y}} \lambda_{\min}(\sum_c y_c \cdot \mathbf{A}^{(c)})K$. Then, $\max_c \mathbf{A}^{(c)} \bullet \mathbf{X}^* =\lambda_{\min}(\sum_c y^*_c \cdot \mathbf{A}^{(c)})K=\mathcal{J}$, which proves Theorem \ref{theorem:duality_of_the_canonical_form}.

\section*{Appendix: Canonical Form of the SDP CSP}\label{subsec:format_canonical_form_sdp}
Formatting the canonical form of the SDP CSP in \eqref{eq:prob:csp_X_relaxed} has two steps: 1) we construct the symmetric coefficient matrices, and 2) the matrices are normalized to a spectral norm $1$.

\subsubsection*{Constructing Coefficient Matrices}
We first rewrite the constraint in \eqref{eq:const:X_psd_diag1} as
\begin{equation}
\begin{aligned}
\mathbf{X} \succeq 0 , \ X_{k,k} = 1 ,\ \forall k \ \Leftrightarrow\mathbf{X}\in\mathcal{X}, \  X_{k,k} \leq 1 ,\ \forall k.
\end{aligned}
\end{equation}
Thus, we obtain the domain of $\mathbf{X}$, $\mathcal{X}$, as \eqref{eq:definition:sdp_canonical_form_primal}. Here, we need to rewrite the inequalities further, $X_{k,k} \leq 1$, $\forall k$, in a form of $\mathbf{A}'\bullet\mathbf{X}\leq b$. By defining the following symmetric matrix,
\begin{equation}
\begin{aligned}
\mathbf{D}^{(k)} \triangleq [D^{(k)}_{i,j}\big|D^{(k)}_{k,k} = 1;D^{(k)}_{i,j} =  0,\ \forall (i,j)\neq(k,k) ],\ \forall k,
\end{aligned}
\end{equation}
$X_{k,k} \leq 1$ can be rewritten using this matrix as 
\begin{equation}\label{eq:const:ax<b:normalized_vector}
\begin{aligned}
\mathbf{D}^{(k)} \bullet \mathbf{X} \leq 1 , \ \forall k .
\end{aligned}
\end{equation}
Next, note that $X_{k,k'}$ and $X_{k',k}$ $\forall k\neq k'$ both describe how likely users $k$ and $k'$ should be assigned in the same slot. 
Thus, the component $aX_{k,k'}$ in the constraints can be rewritten into two components as $\frac{a}{2}X_{k,k'} + \frac{a}{2}X_{k',k}$, where $a$ is a coefficient on $X_{k,k'}$ in a constraint. Based on this fact, we rewrite the coefficients in the constraint \eqref{eq:const:same_bs_diff_slot_X} using a symmetric matrix as
\begin{equation}
\begin{aligned}
\mathbf{F}^{(e)} & \triangleq[F^{(e)}_{i,j}| F^{(e)}_{k_e,k'_e} = F^{(e)}_{k'_e,k_e}= \frac{1}{2}; \\
& F^{(e)}_{i,j}= 0,\forall \{i,j\}\neq\{k'_e,k_e\}] , \forall e, \{k'_e,k_e\} \in \mathcal{E}^{\mathrm{asso}} ,
\end{aligned}
\end{equation}
where $\{k'_e,k_e\}$ is the $e$-th edge in $\mathcal{E}^{\mathrm{asso}}$, i.e., $Q_{k'_e,k_e} = 1$. Then, \eqref{eq:const:same_bs_diff_slot_X} is rewritten as
\begin{equation}\label{eq:const:ax<b:association}
\begin{aligned}
\mathbf{F}^{(e)} \bullet \mathbf{X} \leq -\frac{1}{Z-1},\ \forall \{k'_e,k_e\} \in \mathcal{E}^{\mathrm{asso}}  .
\end{aligned}
\end{equation}
Similarly, for the constraint \eqref{eq:const:max_interference_X}, we define the symmetric matrix that rewrites the coefficients as
\begin{equation}
\begin{aligned}
\mathbf{H}^{(k)} &\triangleq[\mathbf{H}^{(k)}_{i,j}| H^{(k)}_{k',k} = H^{(k)}_{k,k'} = \frac{Z-1}{2Z} S_{k',k}, \forall k'\neq k; \\
& H^{(k)}_{k,k} = 0; H^{(k)}_{i,j} =  0, \forall (i,j) , i\neq k \lor j\neq k  ] , \forall k  , 
\end{aligned}
\end{equation}
which further rewrites the constraint as
\begin{equation}\label{eq:const:ax<b:interference_power}
\begin{aligned}
\mathbf{H}^{(k)} \bullet \mathbf{X} \leq \alpha - \frac{1}{Z}\sum_{k'\neq k} S_{k',k},\ \forall k .
\end{aligned}
\end{equation}

\subsubsection*{Normalizing Coefficient Matrices}
The above constraints in \eqref{eq:const:ax<b:normalized_vector}, \eqref{eq:const:ax<b:association} and \eqref{eq:const:ax<b:interference_power} are formatted in a form $\mathbf{A}'\bullet\mathbf{X}\leq b$ for a matrix $\mathbf{A}'$ and a constant $b$. Note that we aim to format the constraints in the canonical form as $\mathbf{A}\bullet\mathbf{X}\leq 0$, where $\mathbf{A}$ has norm $1$ and the right-hand side is $0$.
To achieve this, we can rewrite $b$ as $b=b/K\cdot\mathbb{I}^{K}\bullet\mathbf{X}$ since the trace of $\mathbf{X}\in\mathcal{X}$ is $K$. Then, we move $b/K\cdot\mathbb{I}^{K}\bullet\mathbf{X}$ to the left-hand side of the constraint and format it as $(\mathbf{A}'-b/K\cdot\mathbb{I}^{K})\bullet\mathbf{X}\leq 0$, where $(\mathbf{A}'-b/K\cdot\mathbb{I}^{K})$ is the coefficient matrices \cite{steurer2010fast}. Finally, by normalizing the matrices to set their spectral norm as $1$, we obtain the constraints in the canonical form.
Specifically, the canonical form of the relaxed SDP CSP in
\eqref{eq:prob:csp_X_relaxed} has coefficient matrices as
\begin{equation}\label{eq:canonical_form_matrices}
\begin{aligned}
&\mathbf{A}^{(k)} = \nsl(\mathbf{D}^{(k)} - \frac{1}{K}\mathbb{I}^{K}), k=1,\dots,K, \\
&\mathbf{A}^{(K+e)} = \nsl(\mathbf{F}^{(e)} + \frac{1}{K(Z-1)}\mathbb{I}^{K}), e=1,\dots,|\mathcal{E}^{\mathrm{asso}}|, \\
&\mathbf{A}^{(K+|\mathcal{E}^{\mathrm{asso}}|+k)} = \nsl( \mathbf{H}^{(k)} - ( \frac{1}{K}\alpha - \frac{1}{KZ}\sum_{k'\neq k} S_{k',k} ) \mathbb{I}^{K}), \\ &\qquad\qquad\qquad k=1,\dots,K,
\end{aligned}
\end{equation}
where $\nsl(\cdot)$ is the normalization function to divide the input matrix with its spectral norm, i.e., $\nsl(\mathbf{A}')=\mathbf{A}'/\|\mathbf{A}\|$. 
The norms of the above matrices in the inputs of $\nsl(\cdot)$ are
\begin{equation}\label{eq:norms_of_canonical_form_matrices}
\begin{aligned}
&\|\mathbf{D}^{(k)} - \frac{1}{K}\mathbb{I}^{K}\| = 1-\frac{1}{K} ,  \ \forall k ,\\
&\|\mathbf{F}^{(e)} + \frac{1}{K(Z-1)}\mathbb{I}^{K}\| = \frac{1}{K(Z-1)}+\frac{1}{2}, \ \forall e ,\\
&\| \mathbf{H}^{(k)} - ( \frac{1}{K}\alpha - \frac{1}{KZ}\sum_{k'\neq k} S_{k',k} ) \mathbb{I}^{K} \| \\
&\ = |\frac{1}{K}\alpha - \frac{1}{KZ}\sum_{k'\neq k} S_{k',k}| + \frac{Z-1}{2Z} \big(\sum_{k'\neq k} ( S_{k',k})^2  \big)^{\frac{1}{2}}, \forall k  .
\end{aligned}
\end{equation}
These expressions are derived by a symbolic mathematics tool \cite{eigenvalues} that computes the eigenvalues of the matrices.
The dual problem is directly formatted using the above matrices in \eqref{eq:canonical_form_matrices}.

\section*{Appendix: The Proof of Theorem \ref{theorem:mmw_convergence}}

Using the regret bound of the general matrix multiplicative weights \cite{arora2007combinatorial}, we have the upper bound of the losses as
\begin{equation}
\begin{aligned}
\sum_{n=1}^{N} \mathbf{L}^{[n]} \bullet \mathbf{X}^{[n]} \leq  \lambda_{min} (\sum_{n=1}^{N} \mathbf{L}^{[n]} ) K  + \eta N K + \frac{\ln K}{\eta} K .  
\end{aligned}
\end{equation}
On the other hand, each constraint's violation, $\mathbf{A}^{c}\bullet\mathbf{X}^{[n]}$, $\forall c,n$, is bounded in $[-K,K]$ since the norm of $\mathbf{A}^{c}$ is $1$ and $\mathbf{X}^{[n]}\in \mathcal{X}$. 
Thus, we can have the lower bound of the losses by applying the hedge rule's regret bound \cite{freund1997decision} as
\begin{equation}
\begin{aligned}
\sum_{n=1}^{N} \mathbf{L}^{[n]} \bullet \mathbf{X}^{[n]} \geq  \max_c \sum_{n=1}^{N} \mathbf{A}^{(c)}\bullet\mathbf{X}^{[n]} - \eta N K - \frac{\ln C}{\eta} K,
\end{aligned}
\end{equation}
where each constraint's violation in each turn is the reward of an expert in the hedge rule. 
By combining the lower and upper bounds, we have
\begin{equation}
\begin{aligned}
&\max_c \sum_{n=1}^{N} \mathbf{A}^{(c)}\bullet\mathbf{X}^{[n]} - \eta N K - \frac{\ln C}{\eta} K \\
&\leq  \lambda_{min} (\sum_{n=1}^{N} \mathbf{L}^{[n]} ) K  + \eta N K + \frac{\ln K}{\eta} K.  
\end{aligned}
\end{equation}
Note that $\sum_{n=1}^{N} \mathbf{L}^{[n]}/N = \sum_c \Bar{y}_c \mathbf{A}^{(c)}$ and  $\max_c \sum_{n=1}^{N} \mathbf{A}^{(c)}\bullet\mathbf{X}^{[n]}/N = \max_c\mathbf{A}^{(c)}\bullet \Bar{\mathbf{X}}$.
By setting $N=\frac{1}{\eta^2}(\ln K + \ln C)$ and substituting the duality gap in the above, we have
\begin{equation}
\begin{aligned}
\mathrm{gap}(\Bar{\mathbf{y}},\Bar{\mathbf{X}}) 
\leq 2 \eta K + \frac{\ln K+\ln C}{N\eta} K = 3 \eta K ,
\end{aligned}
\end{equation}
which proves the statement.

\section*{Appendix: Applying Learning and Spectral Methods for Interference Management}
Directly applying the learning methods, e.g., using GNNs, can hardly return meaningful slot assignments in graph cut/coloring like interference management problems according to the justification in \cite{loukas2019what,xu*2018how,gu2024graph} and the observations in \cite{gu2024graph}.
This is because the GNNs will aggregate the interference information on the neighboring users in their internal processing. This feature losses the user-pair-wise interference information, i.e., the interference graph structure, and therefore, the GNNs cannot differentiate and separate the interfering user pairs. 
One way to address this issue is to use the GNNs to learn the interference graph structure and then apply the learned structure to the graph cut/coloring \cite{gu2024graph}. Nevertheless, this method \cite{gu2024graph} requires an efficient graph cut/coloring scheme for large networks, such as the one designed in this work.

Using spectral graph methods, a graph cut problem can be approximated as an eigenvalue problem \cite{zha2001spectral,vandam2016new}. Specifically, let $\mathbf{v}_k$ be the relaxed unit vector for user $k$, as mentioned in Section \ref{sec:sdp_relaxation_framework}. Considering an undirected interference graph with symmetric adjacency matrix $\mathbf{A}$, the Laplacian matrix is defined as $\mathbf{L} = \mathbf{D} - \mathbf{A}$, where $\mathbf{D}$ is the diagonal degree matrix with $D_{k,k} = \sum_{k'\neq k} A_{k,k'}$. The sum of removed edge values in the cut is proportional to as the quadratic function $\mathbf{L}  \bullet (\mathbf{V}\mathbf{V}^{\rm T}) =\Tr(\mathbf{L} \mathbf{V} \mathbf{V}^{\rm T}) =\Tr(\mathbf{V}^{\rm T} \mathbf{L} \mathbf{V})$, where $\mathbf{V} = [\mathbf{v}_1,\dots,\mathbf{v}_K]^{\rm T}$ is the matrix collecting the relaxed slot assignment vectors of all users. By applying a proper re-parameterizing, $\mathbf{V}$ has $Z$ orthonormal columns \cite{zha2001spectral,vandam2016new}, i.e., $\mathbf{V}^{\rm T} \mathbf{V} = \mathbb{I}^Z$.
With this constraint, the optimal $\mathbf{V}$'s columns maximizing $\Tr(\mathbf{V}^{\rm T} \mathbf{L} \mathbf{V})$ are the $Z$ largest eigenvectors of $\mathbf{L}$ \cite{fan1949theorem}. Then, by normalizing $\mathbf{V}$ in rows, we can obtain the relaxed unit vectors $\mathbf{v}_k$ for each user $k$.
Note that this method's solution has the same form as the gram matrix of the PSD matrix and is rounded to the integer slot assignments $\mathbf{z}$ as discussed in Section \ref{sec:sdp_relaxation_framework}, i.e., it can replace the SDP solver in the framework.
However, it works on a single graph structure and integrating the interference constraint coefficient matrices into the graph structure is not straightforward.

\bibliography{main}
\bibliographystyle{IEEEtran}